\documentclass{article}
\usepackage[utf8]{inputenc}
\usepackage{amsthm}
\usepackage{amsmath}
\usepackage{amssymb}
\usepackage{mathtools}
\usepackage{geometry}
\usepackage{authblk}
\usepackage{algorithm}
\usepackage{algpseudocode} 
\usepackage{listings}
\usepackage{hyperref}

\newtheorem{definition}{Definition}
\newtheorem{lemma}{Lemma}
\newtheorem{corollary}{Corollary}
\newtheorem{theorem}{Theorem}
\newtheorem{proposition}{Proposition}
\newtheorem{conjecture}{Conjecture}
\theoremstyle{remark}
\newtheorem{remark}{Remark}
\newtheorem{openproblem}{Open Problem}

\DeclareMathOperator{\rank}{Rank}
\newcommand\wt{\mathrm{wt}}
\newcommand\ftwo{\mathbb{F}_{2}}
\newcommand{\F}{\mathbb{F}}
\newcommand{\GL}{\mathrm{GL}}

\newcommand\tr[1]{\textrm{Tr}\left(#1\right)}
\newcommand\img[1]{\mathcal{I}m(#1)}

\renewcommand\footnotemark{}
\title{Trims and Extensions of Quadratic APN Functions\thanks{This version of the article has been accepted for publication, after peer review
but is not the Version of Record and does not reflect post-acceptance improvements, or any
corrections. The Version of Record is available online at: \url{https://doi.org/10.1007/s10623-022-01024-4}.}}
\author[1]{Christof Beierle}
\affil[1]{Ruhr University Bochum, Universit\"atsstra\ss e 150, 44801 Bochum, Germany

firstname.lastname@rub.de}
\author[1]{Gregor Leander}
\author[2]{L\'eo Perrin}
\affil[2]{Inria, 2 rue Simone Iff, 75012, Paris, France

leo.perrin@inria.fr}

\date{}

\begin{document}

\maketitle

\begin{abstract}
In this work, we study functions that can be obtained by restricting a vectorial Boolean function $F \colon \ftwo^n \rightarrow \ftwo^n$ to an affine hyperplane of dimension $n-1$ and then projecting the output to an $n-1$-dimensional space. We show that a multiset of $2 \cdot (2^n-1)^2$ EA-equivalence classes of such restrictions defines an EA-invariant for vectorial Boolean functions on $\ftwo^n$. Further, for all of the known quadratic APN functions in dimension $n < 10$, we determine the restrictions that are also APN. Moreover, we construct 6,368 new quadratic APN functions in dimension eight up to EA-equivalence by extending a quadratic APN function in dimension seven. A special focus of this work is on quadratic APN functions with maximum linearity. In particular, we characterize a quadratic APN function $F \colon \ftwo^n \rightarrow \ftwo^n$ with linearity of $2^{n-1}$ by a property of the ortho-derivative of its restriction to a linear hyperplane. Using the fact that all quadratic APN functions in dimension seven are classified, we are able to obtain a classification of all quadratic 8-bit APN functions with linearity $2^7$ up to EA-equivalence.

\end{abstract}

\section{Introduction}
Let us be given two integers $n,m \in \mathbb{N}$ with $m < n$ and two vectorial Boolean functions $F:\F_2^n \to \F_2^n$ and $G:\F_2^{m} \to \F_2^{m}$. We say that $F$ is an {\emph{extension}} of $G$ and that $G$ is a {\emph{restriction}} of $F$, denoted $G \prec  F$, if there exists an affine injective mapping $\phi:\F_2^{m} \to \F_2^n$ and an affine surjection $\varphi:\F_2^n \to \F_2^{m}$ such that
$G=\varphi \circ F \circ \phi$.
\begin{center}
\includegraphics[width=.3\columnwidth]{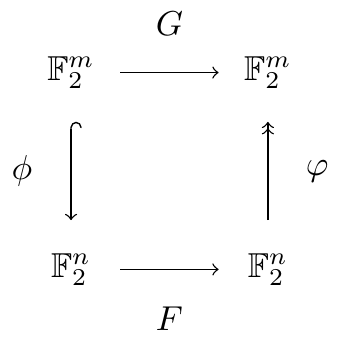}
\end{center}
While the definition is sound for any vectorial Boolean functions $F$ and $G$, we are mainly interested in the case where $F$ and $G$ are both APN functions. 

In the remainder of this work, we concentrate on the case of $m=n-1$, but let us recall first that there are well-known examples for APN restrictions, resp., APN extensions. Indeed, given an APN monomial function $F$ on $\F_{2^{n}}$, then the restriction $G$ of $F$ to any proper subfield of $\F_{2^{n}}$ is APN. 

The notion of restriction defines a strict partial ordering on the set of vectorial Boolean functions with the same dimension in the input and the output. Indeed the relation is irreflexive, as no function is its own extension simply as we exclude the case $n=m$ by requiring $m <n$ . For the same reason, the relation is antisymmetric, i.e. no function can be both an extension and a restriction of any given function. Transitivity follows directly from the definition. 

The notion of restriction for vectorial Boolean functions also defines a notion of restriction on EA-equivalence classes of vectorial Boolean functions. We say that the EA-equivalence class of $F$ is the \emph{extension} of the EA-equivalence class of $G$ if there exist a function $G'$ which is EA-equivalent to $G$ such that $G' \prec F$. This is well-defined as, given two EA-equivalent functions $F$ and $F'$ and a function $G$ that is the restriction of $F$, then there exists an EA-equivalent function $G'$ to $G$ such that $G'$ is the restriction of $F'$. 

As mentioned above, we only focus on the case $m=n-1$, that is extensions and restrictions by adding or removing only a single dimension. One question is how to algorithmically check for a given APN function whether it is an extension or a restriction of another APN function and whether we can discover new APN functions as extensions or restrictions of known functions. In Section~\ref{sec:trimming}, we define an operation called \emph{trimming}, which restricts a vectorial Boolean function $F$ on $\ftwo^n$ to a linear or affine hyperplane of dimension $n-1$ and then projects the output to an $n-1$-dimensional space by discarding one \emph{component} function. Compared to the case in which only a \emph{coordinate} is discarded, we show that this operation is sound in the sense that (1), any function on $\ftwo^{n-1}$ that is a restriction of $F$ is EA-equivalent to a trim of $F$ and, (2), EA-equivalent functions yield EA-equivalent trims when considering all $2 \cdot (2^n-1)^2$ possibilities to choose the affine hyperplanes and the component functions. This way, we obtain an EA-invariant, which we call \emph{trim spectrum}. We then analyze how many of the known quadratic APN functions in dimensions $n < 10$ contain APN restrictions. The results are visualized in a so-called \emph{trimming graph}, see Figures~\ref{fig:trimming-graph-n7} and~\ref{fig:trimming-graph-n8}.

For constructing new APN functions that have a given APN function as a restriction, the focus on EA-equivalence classes is helpful as, in a nutshell, extending a function by one dimension boils down to constructing one additional coordinate function and deducing the values for the remaining input values. More precisely, an APN function $G$ on $\F_2^n$ is a restriction of an APN function in dimension $n+1$ if there exist Boolean functions $r_1,r_2: \F_2^{n} \to \F_2$ and a vectorial Boolean function $G': \F_2^{n} \to \F_2^{n}$ such that 
\begin{eqnarray}
\label{eq:general_form}
	T:
	\begin{cases}
	\ftwo^n \times \ftwo &\to \ftwo^n \times \ftwo \\
	\left(\begin{array}{c}
		x  \\
		y 
	\end{array}\right) &\mapsto
	\left(\begin{array}{c}
		G(x)  \\
		r_1(x) 
	\end{array}\right) \cdot (y+1)+
	\left(\begin{array}{c}
		G'(x)  \\
		r_2(x) 
	\end{array}\right)\cdot y
	\end{cases}
\end{eqnarray}
is APN.
While for very small dimensions it is possible to check if a function $G$ is an APN restriction of some function, this becomes prohibitively expensive as $n$ increases.  However, focusing on quadratic functions  for $T$ and $G$ allows for much stronger results. As we will detail in Section~\ref{sec:r-extendable}, in the case where we want to determine if a given quadratic APN function $G$ is a restriction of a quadratic APN function $T$, the function $G'$ differs from $G$ itself only by a linear mapping $L$ and $r_2$ differs by a linear Boolean function $\ell$ from $r_1$. This, obviously, reduces the search space significantly. By using a recursive tree search with backtracking, we are able to construct 6,368 new quadratic APN functions in dimension $n=8$ up to EA-equivalence. Those new functions are APN extensions of quadratic 7-bit APN functions by construction and it is worth remarking that the vast majority of the previously-known APN functions in dimension $n=8$ are not, i.e., they do not have APN functions within their trim spectrum. As we found out that almost all of the quadratic APN functions in dimension $n=7$ can be extended to quadratic APN functions in dimension $n=8$, we conjecture that for each dimension $n$, there exist APN functions $F_i \colon \ftwo^i \rightarrow \ftwo^i, i=2,\dots,n$ such that $F_2 \prec F_3 \prec \dots \prec F_{n-1} \prec F_n$, so called \emph{recursive APN functions}.

For APN functions in the form of Equation~(\ref{eq:general_form}), the case where $T$ is quadratic and $r_1$ is constant and equal to zero is of particular interest. Indeed, as it was already observed in~\cite[Remark 12]{DBLP:journals/tit/Carlet18}, any quadratic APN function on $\F_2^n$ with linearity $2^{n-1}$, i.e., the highest possible linearity of a quadratic APN function, is EA-equivalent to the extension $T$ of a quadratic APN function $G$ with choosing $r_1 = 0$ (as in the form of Equation~(\ref{eq:general_form})). Moreover, we show that for the ortho-derivative $\pi_G$ of $G$, we have $\langle \pi_G(\alpha), L(\alpha) \rangle=1$ for all $\alpha \in \ftwo^n \setminus \{0\}$ with $\ell(\alpha)=0$. This observation allows us in particular to classify all quadratic APN functions with maximum linearity in dimension eight, simply by using the recent classification of all quadratic 7-bit APN functions into EA-equivalence classes and by recovering $L$ from the ortho-derivative of $G$ using linear algebra. 

In the last part of this paper, we provide some further observations on quadratic APN functions with maximum linearity. In particular, we provide a simple representation of those functions and study their Walsh spectra, as well as their ortho-derivatives.  We conclude by listing several open problems for future work.

\subsection{Related Work}
The effect on the differential uniformity and on the linearity of restricting a vectorial Boolean function to an affine subspace was first studied in~\cite{DBLP:conf/fse/Nyberg94}.
Restrictions of APN functions that are obtained by discarding one output coordinate have also been studied before, see~\cite{Gorodilova+2016+193+202,DBLP:journals/ccds/Idrisova19}.

APN functions are completely classified up to $n \leq 5$, see~\cite{DBLP:journals/dcc/BrinkmannL08}. For $n = 6,7$, we know a complete classification of quadratic and cubic, and quadratic APN functions up to EA-equivalence, respectively~\cite{langevin,DBLP:journals/ccds/Calderini20,DBLP:journals/iacr/KalginI20a}. For $n=8$, at the time of submission of this manuscript in August 2021, we know 26,524 distinct quadratic APN functions up to EA-equivalence. Those are the ones listed by Edel and Pott~\cite{DBLP:journals/amco/EdelP09} in 2009, the 10 functions constructed in~\cite{weng2013quadratic}, the 8,157 functions found by the QAM approach in 2014~\cite{DBLP:journals/dcc/YuWL14,DBLP:journals/iacr/YuWL13}, the two functions coming from the Taniguchi family~\cite{DBLP:journals/dcc/Taniguchi19a}, the 12,921 functions found by the recursive tree search utilizing linear self-equivalences~\cite{DBLP:journals/corr/abs-2009-07204}, and the 5,412 functions found by the QAM approach very recently~\cite{DBLP:journals/iacr/YuP21}. The authors of~\cite{DBLP:journals/iacr/YuP21} conjectured that there are more than 50,000 distinct quadratic APN functions in dimension $n=8$ up to EA-equivalence. 

One interesting result reported in~\cite{DBLP:journals/corr/abs-2009-07204} is the fact that, among the 12,921 found APN functions in dimension $n=8$, four of them have a linearity of $2^{n-1}$, which is the highest value that can possibly be achieved for quadratic APN functions. Note that for $n \leq 4$, every quadratic APN function admits a linearity of $2^{n-1}$ trivially.\footnote{Since there exist linear APN functions in dimension $n \leq 2$, the linearity can be $2^n$ in those cases.} In odd dimension $n \geq 5$, a quadratic APN function cannot have linearity $2^{n-1}$, since every such function must be almost bent~\cite{DBLP:journals/dcc/CarletCZ98}. In dimension $n=6$, there is exactly one quadratic APN function up to EA-equivalence which admits the highest possible linearity of $2^{n-1}$. The existence of quadratic APN functions in dimension $n$ having linearity $2^{n-1}$ is still unknown for $n>8$. Preliminary observations on such functions have been remarked in~\cite{DBLP:journals/tit/Carlet18}.

In~\cite{DBLP:journals/iacr/KalginI20a}, the authors proposed a secondary approach to search for quadratic APN functions in dimension $n+1$ by extending a quadratic APN function in dimension $n$. The approach was to guess the $(n+1)$-th coordinate function and to utilize necessary properties of the algebraic normal form of the extended function. The search was then quite similar to the QAM approach~\cite{DBLP:journals/dcc/YuWL14}. However, no results were reported for $n \geq 7$.

\section{Notation and Preliminaries}
By $\mathbb{N}$, we denote the set of natural numbers $\{1,2,3,\dots\}$ and by $\mathbb{F}_q$, we denote the finite field with $q$ elements. In this work, we focus on functions between finite-dimensional $\ftwo$-vector spaces, also called \emph{vectorial Boolean functions}. For $n \in \mathbb{N}$, the set of invertible linear automorphisms from $\ftwo^n$ to itself is denoted by $\GL(n,\ftwo)$ and in our notation, we use matrices and their corresponding linear mappings interchangeably. In the following, let $\mathbb{V},\mathbb{V}',\mathbb{W},\mathbb{W}'$ be finite-dimensional (non-zero) vector spaces over $\ftwo$. We denote the set of all linear mappings from $\mathbb{V}$ to $\mathbb{W}$ by $\mathcal{L}(\mathbb{V},\mathbb{W})$.  An \emph{affine function} from $\mathbb{V}$ to $\mathbb{W}$ is a function of the form $x \mapsto L(x) + c$, where $L \in \mathcal{L}(\mathbb{V},\mathbb{W})$ and $c \in \mathbb{W}$. Two vectorial Boolean functions $F \colon \mathbb{V} \rightarrow \mathbb{W}, G \colon \mathbb{V}' \rightarrow \mathbb{W}'$ are called \emph{extended-affine equivalent} (\emph{EA-equivalent} for short) if there exist affine bijections $A \colon \mathbb{V}' \rightarrow \mathbb{V}, B \colon \mathbb{W} \rightarrow \mathbb{W}'$ and an affine function $C \colon \mathbb{V}' \rightarrow \mathbb{W}'$ such that $G = B \circ F \circ A + C$. If $C = 0$, the functions $F$ and $G$ are called \emph{affine-equivalent}, and if $C=0$ and $A$ and $B$ are linear bijections, the two functions $F$ and $G$ are called \emph{linear-equivalent}. We are interested in vectorial Boolean functions only up to EA-equivalence since the most important cryptographic properties are invariant under this equivalence relation, most importantly the algebraic degree, the differential uniformity, as well as the linearity. We recall the definition of those three notions in the following paragraphs. For a comprehensive textbook on the theory of Boolean functions and vectorial Boolean functions, we refer to~\cite{carlet_2021}.

Any Boolean function $f \colon \ftwo^n \rightarrow \ftwo$ can be uniquely expressed as a multivariate polynomial in $\ftwo[X_1,\dots,X_n]/(X_1^{2}+X_1,\dots,X_n^{2}+X_n)$ via \begin{equation*}f \colon \ftwo^n \rightarrow \ftwo, \quad x \mapsto \sum_{u \in \ftwo^n}\left(a_u\prod_{i=1}^n x_i^{u_i}\right), \quad a_u \in \ftwo.\end{equation*}
The \emph{algebraic degree} of $f$ is defined as $\max_{\{u \in \ftwo^n \mid a_u \neq 0\}}  \wt(u)$, where $\wt(u)$ denotes the Hamming weight of the binary vector $u$.
Without loss of generality, a function from $\mathbb{V}$ to $\mathbb{W}$ with $\dim(\mathbb{V})=n$, $\dim(\mathbb{W})=m$ can be represented as a function $F \colon \ftwo^n \rightarrow \ftwo^m$ by applying a linear function in the input and in the output. In this representation, $F$ can be given by its $m$ \emph{coordinate functions} $f_1,\dots,f_m \colon \ftwo^n \rightarrow \ftwo$ as $F(x_1,\dots,x_n) = (f_1(x_1,\dots,x_n),\dots,f_m(x_1,\dots,x_n))$. 
  The \emph{algebraic degree} of a vectorial Boolean function from $\mathbb{V}$ to $\mathbb{W}$ is defined as the maximum algebraic degree over all its coordinate functions when represented as a function from $\ftwo^n$ to $\ftwo^m$. Functions with algebraic degree equal to 2 are called \emph{quadratic} and functions with algebraic degree at most 1 are called \emph{affine}.

The \emph{differential uniformity}~\cite{DBLP:conf/eurocrypt/Nyberg93} of a function $F \colon \mathbb{V} \rightarrow \mathbb{W}$ is defined as $\max_{\alpha \in \mathbb{V} \setminus \{0\},\beta \in \mathbb{W}}\lvert\{x \in \mathbb{V} \mid F(x) + F(x+\alpha)=\beta\}\rvert$. It is straightforward to observe that the differential uniformity of a function $F \colon \mathbb{V} \rightarrow \mathbb{W}$ is an even integer larger than or equal to 2. For $\ftwo$-vector spaces $\mathbb{V}, \mathbb{W}$ of the same finite dimension, a function $F \colon \mathbb{V} \rightarrow \mathbb{W}$ with the least possible differential uniformity of 2 is called \emph{almost perfect nonlinear} (or \emph{APN} for short)~\cite{DBLP:conf/crypto/NybergK92}. 

Let $\langle \cdot, \cdot \rangle_{\mathbb{V}} \colon \mathbb{V} \times \mathbb{V} \rightarrow \ftwo$ and $\langle \cdot, \cdot \rangle_{\mathbb{W}} \colon \mathbb{W} \times \mathbb{W} \rightarrow \ftwo$ be non-degenerate symmetric bilinear forms and let $F \colon \mathbb{V} \rightarrow \mathbb{W}$. A \emph{component} of $F$ is a function $\mathbb{V} \rightarrow \ftwo, x \mapsto \langle b,F(x) \rangle_{\mathbb{W}}$, where $b \in \mathbb{W} \setminus \{0\}$.  The \emph{Walsh transform} of $F \colon \mathbb{V} \rightarrow \mathbb{W}$ at point $(\alpha,\beta) \in \mathbb{V} \times \left(\mathbb{W} \setminus \{0\}\right)$ is defined as 
\begin{align*} \widehat{F}_\beta(\alpha) \coloneqq \sum_{x \in \mathbb{V}}(-1)^{\langle \alpha,x\rangle_{\mathbb{V}}+\langle \beta,F(x)\rangle_{\mathbb{W}}}\end{align*}
and the \emph{linearity} of $F$ corresponds to the maximum absolute value of its Walsh transform, i.e., $\max_{\alpha \in \mathbb{V},\beta \in \mathbb{W}\setminus \{0\}} \lvert \widehat{F}_\beta(\alpha) \rvert$. The linearity is a measure on how well a component can be approximated by an affine function. When $\mathbb{V} = \ftwo^n$ for an integer $n \in \mathbb{N}$, we use $\langle x,y \rangle_{\mathbb{V}} = \langle x, y \rangle = \wt (\sum_{i=1}^n x_iy_i) \mod 2$, where $x = (x_1,x_2,\dots,x_n), y = (y_1,y_2,\dots,y_n) \in \ftwo^n$ and $\wt(k)$ denotes the Hamming weight of the binary expansion of $k \in \mathbb{N} \cup \{0\}$. In case of $\mathbb{V} = \mathbb{F}_{2^n}$ for an integer $n \in \mathbb{N}$, we use $\langle x,y \rangle_{\mathbb{V}} = \tr{xy}$, where $x,y \in \mathbb{F}_{2^n}$ and $\mathrm{Tr}$ denotes the absolute trace function, defined as
\begin{align*} \tr{x} = x^{2^0} + x^{2^1} + x^{2^2} + \dots + x^{2^{n-1}}\;.\end{align*}

It is well known that the Walsh transform of a quadratic Boolean function $f \colon \ftwo^n \rightarrow \ftwo$ only takes values in $\{ 0,\pm 2^{\frac{n+k}{2}}\}$, where $k$ is the dimension of the vector space $\{ a \in \ftwo^n \mid  x \mapsto f(x) + f(x+a) \text{ is constant}\}$, see~\cite[Prop.\@ 55]{carlet_2021}. Therefore, the Walsh transform of a quadratic vectorial Boolean function $F \colon \mathbb{V} \rightarrow \mathbb{W}$ can only take 0 or powers of 2 as absolute values. In case that $F$ is APN and $\dim(\mathbb{V}) = \dim(\mathbb{W}) = n >2$, we know that the linearity of $F$ cannot be equal to $2^n$, see~\cite[Prop.\@ 161]{carlet_2021}. Therefore, when $F$ is APN and quadratic, the linearity of $F$ can be at most $2^{n-1}$, which motivates the following definition.

\begin{definition}
  Let $n \in \mathbb{N}, n>2$ and let $\mathbb{V},\mathbb{W}$ be $n$-dimensional $\ftwo$-vector spaces. We say that a quadratic APN function $F \colon \mathbb{V} \rightarrow \mathbb{W}$ has \emph{maximum linearity} if it has linearity $2^{n-1}$. 
\end{definition}

If $n \in \mathbb{N}$ is odd, every quadratic APN function is almost bent, i.e., its Walsh transform can only take values in $\{ 0, \pm, 2^{\frac{n+1}{2}}\}$, see~\cite{DBLP:journals/dcc/CarletCZ98}. Therefore, quadratic APN functions with maximum linearity can only exist in even dimensions $n$ or for $n=3$. Until recently, we only knew the existence of quadratic APN functions with maximum linearity up to $n=6$, see~\cite{DBLP:journals/amco/EdelP09} for an example in dimension 6. In~\cite{DBLP:journals/corr/abs-2009-07204}, the authors found 4 EA-inequivalent instances in dimension $n=8$. It is still an open problem whether quadratic APN functions with maximum linearity exist for $n \geq 10$. 

In this work, we need the notion of the ortho-derivative of a quadratic APN function, which is defined as follows.
\begin{definition}[\cite{ortho_paper}]
Let $\mathbb{V}, \mathbb{W}$ be $\ftwo$-vector spaces of the same finite dimension and let $\langle \cdot, \cdot \rangle_{\mathbb{W}} \colon \mathbb{W} \times \mathbb{W} \rightarrow \ftwo$ be a non-degenerate symmetric bilinear form. Let $G \colon \mathbb{V} \rightarrow \mathbb{W}$ be a quadratic APN function. The \emph{ortho-derivative} of $G$ is defined as the unique function $\pi_G \colon \mathbb{V} \rightarrow \mathbb{W}$ with $\pi_G(0) = 0$ such that, for all $\alpha \in \mathbb{V} \setminus \{0\}$, we have $\pi_G(\alpha) \neq 0$ and
\begin{align*} \forall x \in \mathbb{V} \colon \langle \pi_G(\alpha), B_{\alpha}(x) \rangle_{\mathbb{W}} = 0\;,\end{align*}
where $B_{\alpha} \colon \mathbb{V} \rightarrow \mathbb{W}, x \mapsto G(x) + G(x+\alpha) + G(\alpha) + G(0)$. 
\end{definition}

The ortho-derivative $\pi_G$ of a given quadratic APN function $G : \ftwo^n \to \ftwo^n$ is computed from the difference distribution table (DDT) of $G$ by first identifying, for every $a \in \ftwo^n$, the linear part of the affine space $\{b \in \ftwo^n \mid \mathsf{DDT}_G[a,b] = 2 \}$, and then finding the value $\pi_G(a)$ which is orthogonal to it.\footnote{We recall that $\mathsf{DDT}_G[a,b]$ is defined as $\lvert \{x \in \ftwo^n \mid G(x) + G(x+a) = b\}\rvert$.} Overall, the most computationally demanding step is the computation of the DDT, which takes time $\mathcal{O}(2^{2n})$.

The authors of~\cite{ortho_paper} showed that for two EA-equivalent quadratic APN functions $G,G' \colon \ftwo^n \rightarrow \ftwo^n$, the ortho-derivatives $\pi_G$ and $\pi_{G'}$ are linear-equivalent. Indeed, the linear-equivalence of the ortho-derivatives is a strongly discriminating EA-invariant for quadratic APN functions, which allows in many cases to efficiently detect the EA-inequivalence of two quadratic APN functions. 

\section{Function Trimming}
\label{sec:trimming}
For a given function $F\colon \ftwo^n \rightarrow \ftwo^n$, we want to consider \emph{all} the restrictions $G \prec F$ with $G \colon \ftwo^{n-1} \rightarrow \ftwo^{n-1}$. Since there are many choices for the affine injective mappings  $\phi \colon \ftwo^{n-1} \rightarrow \ftwo^n$ and the affine surjections $\varphi \colon \ftwo^n \rightarrow \ftwo^{n-1}$, this would quickly become infeasible already for very small values of $n$. However, as we are only interested in vectorial Boolean functions up to EA-equivalence, the particular restrictions and projections that have to be taken into account can be significantly reduced, as we outline in the following.  For an element $\alpha \in \ftwo^n \setminus \{0\}$, we denote by $\alpha^{\perp}$ its orthogonal, i.e., $\alpha^{\perp} = \{x \in \ftwo^n \mid \langle \alpha,x \rangle = 0\}$, which is an $n-1$-dimensional linear hyperplane. We denote its complement set $\ftwo^n \setminus \alpha^{\perp}$ by $\overline{\alpha^{\perp}}$. 

For a non-zero element $\beta \in \ftwo^n$ and an element $\gamma \in \ftwo^n$ such that $\langle \beta, \gamma \rangle = 1$, we define $\rho_{\beta}^{(\gamma)}$ as the function
\begin{equation*} \rho_{\beta}^{(\gamma)} \colon \ftwo^n \rightarrow \gamma^{\perp}, \quad x \mapsto x + \beta \cdot \langle \gamma,x\rangle\;.\end{equation*}
If we represent $\ftwo^n$ as the direct sum $\ftwo^n = \gamma^{\perp} \oplus \{0, \beta\}$, for any vector $x = x_{\beta} \oplus x_{\gamma} \in \ftwo^n$ with $x_{\gamma} \in \gamma^{\perp}, x_{\beta} \in \{0,\beta\}$, we have $\rho_{\beta}^{(\gamma)}(x) = x_{\gamma}$, i.e, $\rho_{\beta}^{(\gamma)}$ is a projection of $\ftwo^n$ to $\gamma^{\perp}$. The following definition precisely captures the idea of restricting the input and projecting the output of a vectorial Boolean function.

\begin{definition}[Trim along $(H,\beta)$]
Let $n \in \mathbb{N}, n \geq 2$, $\beta \in \ftwo^n$ be a non-zero element and let $H \subseteq \ftwo^n$ be a hyperplane of dimension $n-1$, so that $H = \alpha^\perp$ or $H = \overline{\alpha^\perp}$ for some non-zero $\alpha \in \ftwo^n$.
Let $\epsilon \in \ftwo^n$ be zero if $H = \alpha^{\perp}$ and $\epsilon \notin \alpha^{\perp}$ otherwise, and let $\gamma \in \ftwo^n \setminus \beta^{\perp}$. 
The \emph{trim of a function $F$ along $(H,\beta)$ with respect to $\epsilon,\gamma$} is then defined as
\begin{align*}
    \mathcal{T}_{H \leadsto \beta}^{\epsilon,\gamma}F \colon \alpha^{\perp} &\rightarrow \gamma^{\perp} \\
    x &\mapsto \rho_{\beta}^{(\gamma)} \circ F(x + \epsilon) = \begin{cases}
    F(x + \epsilon) &\text{if } F(x+\epsilon) \in \gamma^{\perp} \\
    F(x + \epsilon) + \beta &\text{otherwise}
    \end{cases}~.
\end{align*}
\end{definition}

Note that, for a given $F\colon \ftwo^n \rightarrow \ftwo^n$, we can convert every function $\mathcal{T}_{H \leadsto \beta}^{\epsilon,\gamma}F\colon \alpha^{\perp} \rightarrow \gamma^{\perp}$ to a function from $\ftwo^{n-1}$ to $\ftwo^{n-1}$ by applying linear bijections in the input and in the output. Moreover, if we are only interested in $\mathcal{T}_{H \leadsto \beta}^{\epsilon,\gamma}F$ up to affine-equivalence, the particular choice of $\epsilon$ and $\gamma$ does not matter, as we show in the following proposition.

\begin{proposition}
  \label{prop:epsilon-doesnt-matter}
 Let $n \in \mathbb{N}, n \geq 2$. For a fixed choice of $(H,\beta)$ with $\beta \in \ftwo^n$ being a non-zero element and  $H \subseteq \ftwo^n$ being a hyperplane of dimension $n-1$, all trims of $F$ along $(H, \beta)$ with respect to some $\epsilon,\gamma$ are affine-equivalent.
\end{proposition}
\begin{proof}
We first consider the case of a fixed $\gamma \in \ftwo^n\setminus \beta^{\perp}$ and  varying $\epsilon$.
  If $H$ is a linear hyperplane, i.e., $H = \alpha^{\perp}$ for some non-zero $\alpha \in \ftwo^n$, there is only one valid choice of $\epsilon$, i.e., $\epsilon = 0$. Let us therefore consider the case that $H = \overline{\alpha^\perp}$. Let $\epsilon, \epsilon' \in \ftwo^n \setminus \alpha^{\perp}$ and let $\epsilon_{\alpha} \in \alpha^{\perp}$ be such that $\epsilon' = \epsilon + \epsilon_{\alpha}$.  We have
  \[ \mathcal{T}_{H \leadsto \beta}^{\epsilon',\gamma}F(x) = F((x + \epsilon_{\alpha}) + \epsilon) + \beta \langle \gamma, F((x + \epsilon_{\alpha}) + \epsilon) \rangle = \mathcal{T}_{H \leadsto \beta}^{\epsilon,\gamma}F(x+ \epsilon_{\alpha})\;,\]
  so $\mathcal{T}_{H \leadsto \beta}^{\epsilon',\gamma}F$ is affine-equivalent to $\mathcal{T}_{H \leadsto \beta}^{\epsilon,\gamma}F$.
  
  Let us now consider the case of a fixed $\epsilon$ and varying $\gamma \in \ftwo^n \setminus \beta^{\perp}$. Let $\gamma, \gamma' \in \ftwo^n \setminus \beta^{\perp}$ and let $Q$ be the linear isomorphism $Q \colon \gamma^{\perp} \rightarrow {\gamma'}^{\perp}, x \mapsto x + \beta \langle \gamma',x\rangle$. We then have
  \begin{align*}
      Q \circ \mathcal{T}_{H \leadsto \beta}^{\epsilon,\gamma}F(x) &= F(x+ \epsilon) + \beta \langle \gamma,F(x + \epsilon)\rangle + \beta \big\langle \gamma',F(x+\epsilon) + \beta \langle \gamma,F(x+\epsilon) \rangle \big\rangle \\
      &= \begin{cases}
        F(x+\epsilon) + \beta \langle \gamma',F(x+\epsilon) \rangle &\text{if } \langle \gamma,F(x+\epsilon) \rangle = 0 \\
        F(x+\epsilon) + \beta \langle \gamma',F(x+\epsilon) \rangle + \beta + \beta\langle \gamma',\beta\rangle &\text{if } \langle \gamma,F(x+\epsilon) \rangle = 1 
      \end{cases} \\
      &= F(x+\epsilon) + \beta \langle \gamma',F(x+\epsilon) \rangle = \mathcal{T}_{H \leadsto \beta}^{\epsilon,\gamma'}F(x)\;,
  \end{align*}
  so $\mathcal{T}_{H \leadsto \beta}^{\epsilon,\gamma'}F$ is linear-equivalent to $\mathcal{T}_{H \leadsto \beta}^{\epsilon,\gamma}F$.
\end{proof}

Because of the above proposition, we say that the EA-equivalence class of $\mathcal{T}^{\epsilon,\gamma}_{H \leadsto \beta}F$ is \emph{the trim of $F$ along $(H,\beta)$}, denoted $\mathcal{T}_{H \leadsto \beta}F$. 

\begin{remark}
To the best of our knowledge, the effect on the differential uniformity and on the linearity of restricting a vectorial Boolean function to an affine subspace was first studied in~\cite{DBLP:conf/fse/Nyberg94}. Moreover, the paper~\cite{DBLP:conf/fse/Nyberg94} also derived an upper bound and a lower bound on the differential uniformity of a vectorial Boolean function when composing it with an affine surjection from the output. There are some works that study subfunctions of APN functions obtained by discarding one output coordinate, e.g.,~\cite{Gorodilova+2016+193+202,DBLP:journals/ccds/Idrisova19}. We remark that the notion of a trim covers the case of restricting the input of $F$ to an affine hyperplane and then discarding one output coordinate of $F$, but it is more general than that. Indeed, if $F \colon \ftwo^n \rightarrow \ftwo^n$ is represented by its coordinate functions $F = (f_1,\dots,f_n)$ with $f_i \colon \ftwo^n \rightarrow \ftwo, i \in \{1,\dots,n\}$, the function $F' \colon \ftwo^{n} \rightarrow \ftwo^{n-1}, F' = (f_1,\dots,f_{j-1},f_{j+1},\dots,f_n)$ that is obtained by discarding the $j$-th coordinate of $F$ is linear-equivalent to the function $\rho_{e_j}^{(e_j)} \circ F \colon \ftwo^n \rightarrow e_j^{\perp}$, where $e_j = (0,\dots,0,1,0,\dots,0)$ is the $j$-th unit vector in $\ftwo^n$.
\end{remark}

The following proposition  states that every restriction of $F$ in dimension $n-1$ is EA-equivalent to a trim of $F$.

\begin{proposition}
\label{prop:reduction_implies_trim}
Let $n \in \mathbb{N}, n \geq 2$ and let $F \colon \ftwo^n \rightarrow \ftwo^n$ and $G \colon \ftwo^{n-1} \rightarrow \ftwo^{n-1}$ be given with $G \prec F$. Then, there exists a non-zero element $\beta \in \ftwo^n$ and an affine hyperplane $H \subseteq \ftwo^n$ of dimension $n-1$ such that $G$ is EA-equivalent to an (every) element in $\mathcal{T}_{H \leadsto \beta}F$.
\end{proposition}
\begin{proof}
By definition, there exists an affine injective mapping $\phi \colon \ftwo^{n-1} \rightarrow \ftwo^{n}$ and an affine surjection $\varphi \colon \ftwo^n \rightarrow \ftwo^{n-1}$ such that $G = \varphi \circ F \circ \phi$. Let $\tilde{\phi} \colon \ftwo^{n-1} \rightarrow \img{\tilde{\phi}}$ denote the linear part of $\phi$, such that $\phi = \tilde{\phi} + \epsilon$ for $\epsilon \in \ftwo^{n}$ with $\epsilon = 0$ if $\phi$ is linear and $\epsilon \notin \img{\tilde{\phi}}$ otherwise. Note that $\tilde{\phi}$ is a bijection from $\ftwo^{n-1}$ to $\img{\tilde{\phi}}$, where $\img{\tilde{\phi}} = \alpha^{\bot}$ for a non-zero element $\alpha \in \ftwo^n$. Further, let $\tilde{\varphi}$ denote the linear part of $\varphi$ such that $\varphi = \tilde{\varphi} + b$ with $b \in \ftwo^{n-1}$. We therefore have, for all $x \in \alpha^{\bot}$, 
\begin{align*} G \circ {\tilde{\phi}}^{-1} (x) = (\tilde{\varphi} \circ F(x + \epsilon)) + b\;.\end{align*}
Let $M$ be the $(n-1) \times n$ matrix over $\ftwo$ such that $\tilde{\varphi}(x) = Mx$ for all $x \in \ftwo^n$. Since $\tilde{\varphi}$ is surjective, the matrix $M$ has rank $n-1$, so there exist $n-1$ linearly independent columns in $M$. Let $P$ be an $n \times n$ permutation matrix such that $P^{-1}$ permutes those linearly independent columns to the left side, i.e, we have
\begin{align*} MP^{-1} = \left(\begin{array}{cc} A^{-1} & \beta \end{array}\right)\;, \end{align*}
where $A$ is an invertible $(n-1) \times (n-1)$ matrix and $\beta$ is a column vector in $\ftwo^{n-1}$. We then have $A M P^{-1} = \left(\begin{array}{cc} I & A\beta \end{array}\right)$, where $I$ denotes the $(n-1) \times (n-1)$ identity matrix. Thus, for all $x \in \alpha^{\bot}$, we have
\begin{align*} \left( \begin{array}{c}A(G \circ {\tilde{\phi}}^{-1} (x)) \\ 0 \end{array} \right) &= \left( \begin{array}{cc} I & A\beta \\ 0 & 0 \end{array}\right) PF(x+ \epsilon) + \left(\begin{array}{c}Ab \\ 0 \end{array} \right) \\
&= P F(x + \epsilon) + \left(\begin{array}{c}A\beta \\ 1 \end{array} \right) \langle e_n,PF(x + \epsilon) \rangle + \left(\begin{array}{c}Ab \\ 0 \end{array} \right)\\
&= P \hspace{-.2em}\left( \hspace{-.2em} F(x + \epsilon) + P^{-1}\left(\begin{array}{c}A\beta \\ 1 \end{array} \right) \langle P^{\top}e_n,F(x + \epsilon) \rangle + P^{-1}\left(\begin{array}{c}A\beta  \\ 0 \end{array} \right)\hspace{-.2em} \right),\end{align*}
and finally 
\begin{align*} P^{-1} \left( \begin{array}{c}A(G \circ {\tilde{\phi}}^{-1} (x)) \\ 0 \end{array} \right) + P^{-1}\left(\begin{array}{c}A\beta  \\ 0 \end{array} \right) &=   F(x + \epsilon) + P^{-1}\left(\begin{array}{c}A\beta \\ 1 \end{array} \right) \langle P^{\top}e_n,F(x + \epsilon) \rangle .\end{align*}
\end{proof}

This proposition allows to reduce the number of restrictions we need to consider for a given function $F \colon \ftwo^n \rightarrow \ftwo^n$ to  $2 \cdot (2^n-1)^2$, i.e., we only need to consider all the trims $\mathcal{T}_{H \leadsto \beta}F$ of $F$, where $(H,\beta)$ takes all possible values.
In the following, we establish the fact that for EA-equivalent functions $F,G \colon \ftwo^n \rightarrow \ftwo^n$, the multiset of all possible trims of $F$ is the same as the multiset of all possible trims of $G$.

\begin{proposition}
\label{prop:ea-trimmings}
Let $n \in \mathbb{N},n\geq 2$ and let $F,G \colon \ftwo^n \rightarrow \ftwo^n$ be EA-equivalent via $G = B \circ F \circ (A + a) + b + C$ with $A,B \in \GL(n,\ftwo)$, $a,b \in \ftwo^n$ and $C$ being an affine function in $\ftwo^n$. Then, for each hyperplane $H \subseteq \ftwo^n$ and each $\beta \in \ftwo^n \setminus \{0 \}$, we have that $\mathcal{T}_{H \leadsto \beta}G$ is EA-equivalent\footnote{More precisely, since $\mathcal{T}_{H \leadsto \beta}G$ and $\mathcal{T}_{H' \leadsto \beta'}F$ are EA-equivalence classes, we have to say that that any representative in $\mathcal{T}_{H \leadsto \beta}G$ is EA-equivalent to any representative in $\mathcal{T}_{H' \leadsto \beta'}F$.} to $\mathcal{T}_{H' \leadsto \beta'}F$, where $H' =A(H)+a $ and $\beta' = B^{-1}(\beta)$.
\end{proposition}
\begin{proof}
Let us fix a non-zero element $\beta \in \ftwo^n$ and an $n-1$-dimensional hyperplane $H \subseteq \ftwo^{n}$ such that $H = \alpha^{\perp}$ or $H = \overline{\alpha^{\perp}}$ for a non-zero $\alpha \in \ftwo^n$. Further, let $\gamma \in \ftwo^n \setminus \beta^{\perp}$ and $\epsilon \in \ftwo^n$ with $\epsilon = 0$ if $H = \alpha^{\perp}$ and $\epsilon \notin \alpha^{\perp}$ if $H = \overline{\alpha^{\perp}}$. The proof is divided in three parts.
\begin{itemize}
\item Let first $ G = F \circ (A + a)$, where $A \in \GL(n,\ftwo)$ and $a \in \ftwo^n$. Let $\alpha' \in \ftwo^{n}$ be such that $A(\alpha^{\perp}) = {\alpha'}^{\perp}$. First, let us consider the case of $H = \alpha^{\perp}$. For any $x \in \alpha^{\perp}$, we then have
\begin{align*}
       \mathcal{T}_{H \leadsto \beta}^{0,\gamma}G(x) &= G (x) + \beta \langle \gamma,G (x )\rangle = F  ( A (x) + a) + \beta \langle \gamma,F ( A (x) + a)\rangle \\
      &= \begin{cases}
      \mathcal{T}_{A(H)+a \leadsto \beta}^{0,\gamma}F(A(x) + a) &\text{ if } a \in A(H) \\
      \mathcal{T}_{A(H)+a \leadsto \beta}^{a ,\gamma}F(A(x)) &\text{ if } a \notin A(H)
      \end{cases}\;.
\end{align*}

Let us now consider the case of $H = \overline{\alpha^{\perp}}$. Then, $H = \alpha^{\perp} + \epsilon$ for $\epsilon \notin \alpha^{\perp}$. For any $x \in \alpha^{\perp}$, we then have
\begin{align*}
       \mathcal{T}_{H \leadsto \beta}^{\epsilon,\gamma}G(x) &= G (x + \epsilon) + \beta \langle \gamma,G (x + \epsilon)\rangle \\
       &= F  ( A (x + \epsilon) + a) + \beta \langle \gamma,F ( A (x + \epsilon) + a)\rangle \\
      &= F  ( A (x) + a + A(\epsilon)) + \beta \langle \gamma,F ( A (x) + a + A(\epsilon))\rangle \\
      &= \begin{cases}
      \mathcal{T}_{A(H)+a \leadsto \beta}^{0,\gamma}F(A(x) + (a+ A(\epsilon))) &\text{ if } a+A(\epsilon) \in {\alpha'}^{\perp} \\
      \mathcal{T}_{A(H)+a \leadsto \beta}^{a + A(\epsilon),\gamma}F(A(x)) &\text{ if } a+A(\epsilon) \notin {\alpha'}^{\perp}
      \end{cases}\;.
\end{align*}

\item Let now $ G = B \circ F + b$, where $B \in \GL(n,\ftwo), b \in \ftwo^n$. Let $\gamma' \coloneqq B^{\top}(\gamma)$. For any $x \in \alpha^{\perp}$, we then have
\begin{align*}
    \mathcal{T}_{H \leadsto \beta}^{\epsilon,\gamma}G(x) &= G (x + \epsilon) + \beta \langle \gamma,G (x + \epsilon)\rangle \\
    &= B(F (x + \epsilon)) + \beta \langle \gamma,B(F( x + \epsilon))\rangle + (b + \beta \langle \gamma,b \rangle)  \\
    &= B\big(F (x + \epsilon) + B^{-1}(\beta) \langle \gamma',F( x + \epsilon)\rangle\big) + (b + \beta \langle \gamma,b \rangle) \\
    &= B \circ \mathcal{T}_{H \leadsto \beta'}^{\epsilon,\gamma'}F(x) + (b + \beta \langle \gamma,b \rangle)\;.
\end{align*}

\item Finally, let $G = F + C$ with $C$ being an affine function in $\ftwo^n$. For any $x \in \alpha^{\perp}$, we then have
\begin{align*}
    \mathcal{T}_{H \leadsto \beta}^{\epsilon,\gamma}G(x) &= G (x + \epsilon) + \beta \langle \gamma,G (x + \epsilon)\rangle \\ &= F (x + \epsilon) + \beta \langle \gamma,F (x + \epsilon)\rangle + C (x + \epsilon) + \beta \langle \gamma,C (x + \epsilon)\rangle \\
    &= \mathcal{T}_{H \leadsto \beta}^{\epsilon,\gamma}F(x) + \left( C (x + \epsilon) + \beta \langle \gamma,C (x + \epsilon)\rangle \right)\;.
\end{align*}
\end{itemize}
\end{proof}

Proposition~\ref{prop:ea-trimmings} motivates us to define the notion of the \emph{trim spectrum} of a function $F \colon \ftwo^n \rightarrow \ftwo^n$, which is an EA-invariant. In the following, let $\mathcal{H}_n$ denote the set of all $n-1$-dimensional hyperplanes of $\ftwo^n$, i.e., $\mathcal{H}_n \coloneqq \{\alpha^{\perp} \mid \alpha \in \ftwo^n \setminus \{0\}\} \cup \{\overline{\alpha^{\perp}} \mid \alpha \in \ftwo^n \setminus \{0\}\}$.

\begin{definition}[Trim Spectrum]
   Let $n \in \mathbb{N},n\geq 2$. The \emph{trim spectrum} of a function $F \colon \ftwo^n \rightarrow \ftwo^n$ is the multiset of all trims of $F$ along $(H,\beta)$, where $(H,\beta)$ takes all $2 \cdot (2^n-1)^2$ possibilities, i.e., the multiset $\{ \mathcal{T}_{H \leadsto \beta}F \mid H \in \mathcal{H}_n, \beta \in \ftwo^n \setminus \{0\} \}$. 
\end{definition}

 We recall that for an $n-1$-dimensional hyperplane $H \subseteq \ftwo^n$ and a non-zero element $\beta \in \ftwo^n$, the trim $\mathcal{T}_{H \leadsto \beta}F$ refers to an EA-equivalence class, and not to a particular vectorial Boolean function.
 
 \begin{corollary}
 The trim spectrum of a vectorial Boolean function is an EA-invariant. In other words, for $n \in \mathbb{N},n\geq 2$, if $F,G \colon \ftwo^n \rightarrow \ftwo^n$ are EA-equivalent, the multisets $\{ \mathcal{T}_{H \leadsto \beta}F \mid H \in \mathcal{H}_n, \beta \in \ftwo^n \setminus \{0\} \}$ and $\{ \mathcal{T}_{H \leadsto \beta}G \mid H \in \mathcal{H}_n, \beta \in \ftwo^n \setminus \{0\} \}$ consist of the same EA-equivalence classes with the same multiplicities.
 \end{corollary}
\begin{proof}
This follows directly from Proposition~\ref{prop:ea-trimmings} and from the fact that the two functions $\mathcal{H}_n \rightarrow \mathcal{H}_n, H \mapsto H' = A(H)+a$ and $\ftwo^n \setminus \{0\} \rightarrow \ftwo^n \setminus \{0\}, \beta \mapsto \beta' = B^{-1}(\beta)$ are permutations for $A,B \in \GL(n,\ftwo), a \in \ftwo^n$.
\end{proof}

\begin{remark}
\label{rem:simplified_trim}
We note that for a quadratic function $F \colon \ftwo^n \rightarrow \ftwo^n$, the trim spectrum can be simplified by only considering the \emph{linear} hyperplanes $H$. In fact, for $\alpha, \beta \in \ftwo^n \setminus \{0 \}$ and $\epsilon \in \overline{\alpha^{\bot}}, \gamma \in \overline{\beta^{\bot}}, x \in \alpha^{\bot}$, we have 
\begin{align*}\mathcal{T}_{\overline{\alpha^{\bot}} \leadsto \beta}^{\epsilon,\gamma}F(x) &= F(x+ \epsilon) + \beta \langle \gamma, F(x + \epsilon)\rangle  \\
&= A_{\epsilon}(x) + \beta \langle \gamma, A_{\epsilon}(x)\rangle + \mathcal{T}_{\alpha^{\bot} \leadsto \beta}^{0,\gamma}F(x)\;,\end{align*}
where $A_{\epsilon}$ is the affine mapping such that, for all $x \in \ftwo^n$, we have  $F(x) + F(x+\epsilon) = A_{\epsilon}(x)$. Thus, the two trims $\mathcal{T}_{\overline{\alpha^{\bot}} \leadsto \beta}^{\epsilon,\gamma}F$ and $\mathcal{T}_{\alpha^{\bot} \leadsto \beta}^{0,\gamma}F$ are EA-equivalent.
\end{remark}

\subsection{APN-Trims of APN Functions in Small Dimension}
We could ask the following question: Among the APN functions that were known before our work, how many have APN functions as a restriction in one dimension lower?

To answer this question for the quadratic APN functions in small dimension, we proceeded as follows.
For each quadratic APN function $F$ in dimension $3 \leq n \leq 7$ up to EA-equivalence, we checked whether the trim spectrum of $F$ contains an APN function. To illustrate this data, we plotted it in a \emph{trimming graph} defined as follows. Each EA-equivalence class of a quadratic APN function corresponds to a node, and each node is at a height corresponding to the value of $n$ ($n=3$ at the lowest level, and then it is incremented going up each level). Then, there is a vertex from a node $F$ at height $n$ down to a node $G$ at height $n-1$ if the function $G$ is in the trim spectrum of $F$. Functions that are neither the start nor the end of an edge are not represented. The result can be seen in the graph given in Figure~\ref{fig:trimming-graph-n7}. This graph is complete in the sense that all EA-equivalence classes of quadratic APN functions in dimension up to $n=7$ are known and we have taken into account all the relevant functions. More precisely, for $n \leq 4$, there is only the EA-equivalence class corresponding to the function $x \mapsto x^3$ over $\F_{2^n}$, for $n = 5$, there are the two EA-equivalence classes corresponding to the functions $x \mapsto x^3$ and $x \mapsto x^5$, for $n=6$, there are the 13 quadratic EA-equivalence classes listed in~\cite[Table 5]{DBLP:journals/amco/EdelP09}, and for $n=7$, there are 488 EA-equivalence classes of quadratic APN functions~\cite{DBLP:journals/iacr/KalginI20a}.

\begin{figure}[htb]
    \centering
    \includegraphics[width=.99\columnwidth]{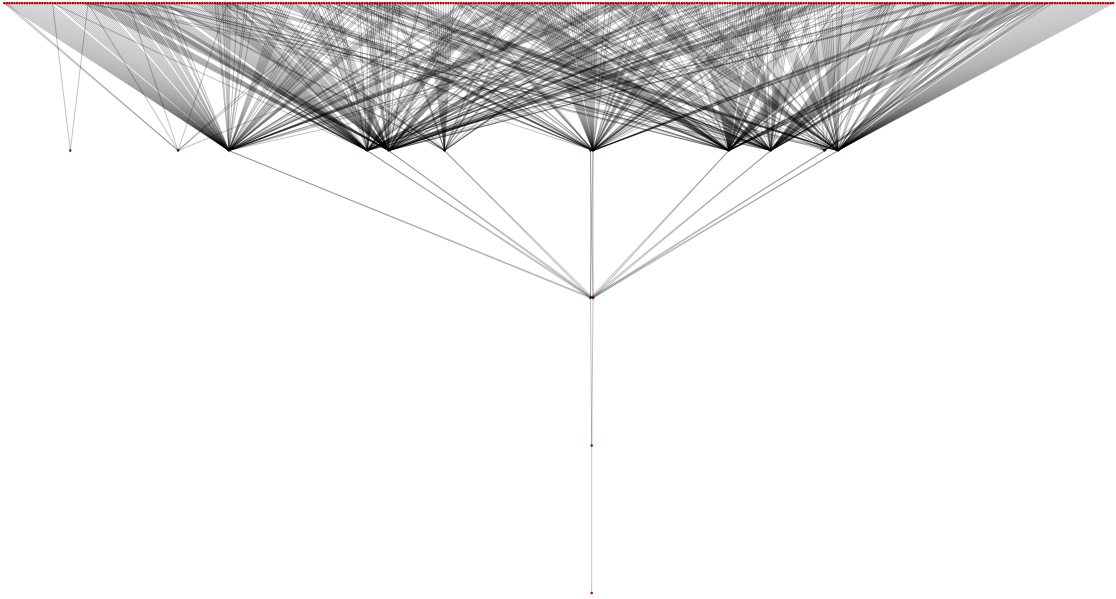}
    \caption{\label{fig:trimming-graph-n7}The trimming graph of all quadratic APN functions for $n \leq 7$. Functions (i.e., nodes) that are neither the start nor the end of an edge are not represented.}
\end{figure}

As we can see, the trimming graph given in Figure~\ref{fig:trimming-graph-n7} has a complex structure where the following properties stand out.
\begin{itemize}
    \item Two EA-equivalence classes of quadratic 6-bit APN functions are in the trim spectra of some 7-bit ones, but do not have any APN functions in their own trim spectra. Those functions correspond to Function no.\@ 1.2 and Function no.\@ 2.1 in~\cite[Table 5]{DBLP:journals/amco/EdelP09} and can be given in univariate representation as $x \mapsto x^3 + g^{11}x^6 + gx^9$ and $x \mapsto x^3 + gx^{24} + x^{10}$, respectively, where $g \in \F_{2^6}$ is an element with minimal polynomial $X^6 + X^4 + X^3 + X + 1$.
    \item Conversely, there is a 6-bit quadratic APN function that is not in the trim spectrum of any 7-bit one, but it has 5-bit APN functions in its trim spectrum. This is Function no.\@ 2.6 in~\cite[Table 5]{DBLP:journals/amco/EdelP09} (see also the representation we provide in Section~\ref{sec:gold}). Note that this observation was already reported in~\cite[Sec.\@ 3]{DBLP:journals/iacr/KalginI20a}.
    \item Most of the EA-equivalence classes of quadratic 6-bit APN functions are in the trim spectra of multiple 7-bit ones, and most of the quadratic 7-bit APN functions that have APN functions in their trim spectra have multiple different EA-equivalence classes of APN functions in their trim spectra.
    \item Several of the quadratic 7-bit APN functions have a unique EA-equivalence class of 6-bit APN functions in their trim spectrum.
    \item For $n=6$, one out of the 13 EA-inequivalent quadratic APN functions does not appear in the graph because it has no APN functions in its trim spectrum, and does not belong to the trim spectrum of any quadratic 7-bit APN function. That function is $x \mapsto x^3$. For $n=7$, 50 functions out of the 488 EA-inequivalent quadratic APN functions do not appear because they do not have any APN functions in their trim spectra.
\end{itemize}

Moreover, for each of the known quadratic APN functions in dimension $n=8$ up to EA-equivalence (except the new ones constructed in Section~\ref{sec:r-extendable}), we also checked whether it contains an APN function within its trim spectrum. The result is given in the graph depicted in Figure~\ref{fig:trimming-graph-n8}. Unlike the previous one, this graph is \emph{not} complete since quadratic APN functions in dimension eight are not classified yet, so we do not have a full list of all the quadratic 8-bit APN functions.

\begin{figure}[htb]
    \centering
    \includegraphics[width=.99\columnwidth]{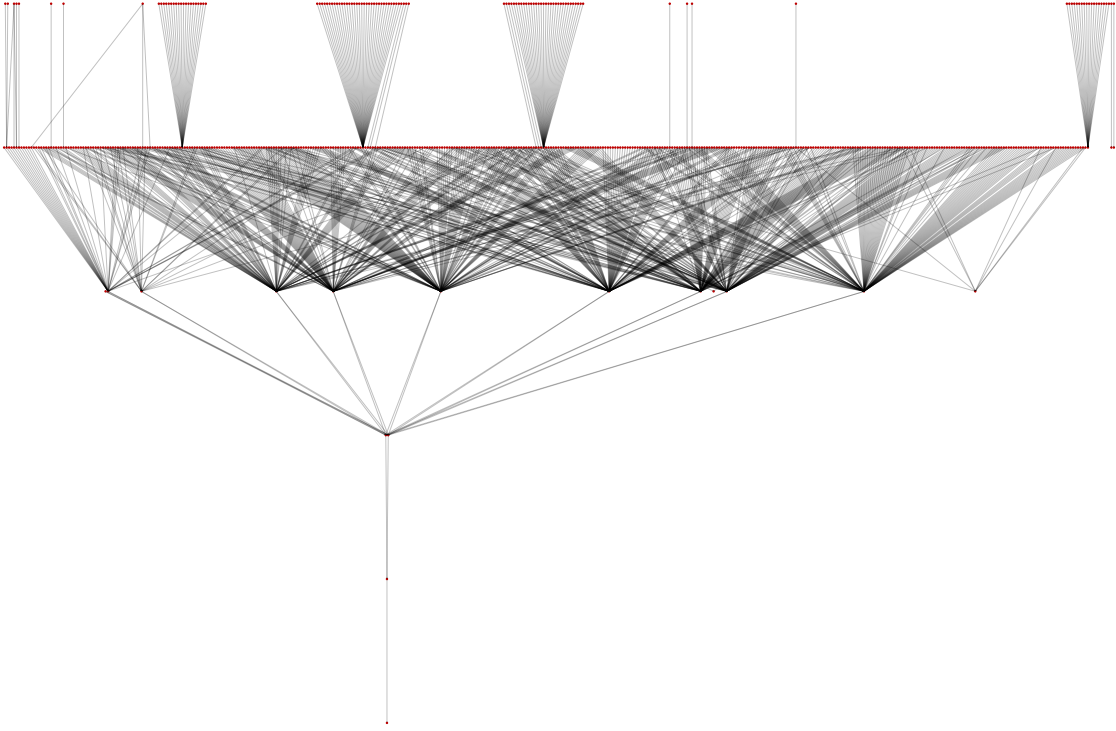}
    \caption{\label{fig:trimming-graph-n8}The trimming graph of all quadratic APN functions in dimension $n$ with $3 \leq n \leq 8$ that were known before our work. Functions (i.e., nodes) that are neither the start nor the end of an edge are not represented.}
\end{figure}

Looking at all of the 26,524 previously known quadratic 8-bit APN functions, we observe the following: 
\begin{itemize}
    \item Only 123 of those functions have an APN function within their trim spectrum.
    \item 15 have a unique EA-equivalence class of 7-bit APN functions in their trim spectrum that is also not in the trim spectrum of another of the 26,524 quadratic 8-bit APN functions. 
    \item Three of the EA-equivalence classes quadratic 7-bit APN functions are in the trim spectrum of quadratic 8-bit APN functions, but do not have any APN functions in their own trim spectrum (meaning that these functions appear in Figure~\ref{fig:trimming-graph-n8} but not in Figure~\ref{fig:trimming-graph-n7}). 
\end{itemize}

To the best of our knowledge, there are 60 known quadratic APN functions in dimension 9 up to EA-equivalence. They either correspond to a polynomial with coefficients in $\ftwo$~\cite{cryptoeprint:2019:1491}, to the (generalized) isotopic shift construction~\cite{DBLP:journals/tit/BudaghyanCCCV20,budaghyan2020generalized}, to the infinite families given in~\cite{DBLP:journals/ffa/BudaghyanCL09,budaghyan2009construction}, or to one of the 35 instances presented in~\cite{DBLP:journals/corr/abs-2009-07204}. None of those 60 instances contains an APN trim within their trim spectrum. 

We further checked for some quadratic APN instances in dimension $n=10$ that come from infinite families (i.e., the instances 10.1--10.2 and 10.5--10.17 from~\cite{apn-list-11}) and for the 5 sporadic instances presented in~\cite{DBLP:journals/corr/abs-2009-07204} whether they contain an APN trim within their trim spectrum. This is not the case for any of those functions.

\paragraph{On the 6-Bit APN Function Inequivalent to a Quadratic Function} From~\cite{DBLP:journals/ccds/Calderini20}, we know that the known 6-bit APN function CCZ-inequivalent to a quadratic function consist of 25 distinct EA-equivalence classes. For each of those EA-equivalence classes, we chose a representative and checked whether the trim spectrum contains APN functions. Indeed, this is the case for 21 out of the 25 classes. For one class, all the APN trims are quadratic, while for the 20 other classes, all the APN trims are cubic. 

\paragraph{On the Non-Quadratic APN Monomial Functions} We further checked for all of the non-quadratic APN monomial functions in dimension $n$ with $5 \leq n \leq 10$ whether the trim spectrum contains an APN function. This is not the case.

In~\cite{trims_code}, we provide the source code for computing the trim spectrum of a given vectorial Boolean function, as well as the source code for computing the trimming graphs as given in Figure~\ref{fig:trimming-graph-n7} and Figure~\ref{fig:trimming-graph-n8}.

\subsection{Recursive APN Functions}
As one can observe from the trimming graph depicted in Figure~\ref{fig:trimming-graph-n8}, there are quadratic APN functions in dimension $n=8$ that contain an APN function in dimension $n=7$ as a restriction, which again contains an APN restriction in dimension $n=6$, and so forth up to $n=2$. This property is indicated by a path from the top level to the bottom level of the trimming graph (note that the trimming graph only depicts the nodes for $n \geq 3$). We call such functions \emph{recursive} APN functions, formally defined in the following.

\begin{definition}
Let $n \in \mathbb{N}, n\geq 2$. An APN function $F_n \colon \ftwo^n \rightarrow \ftwo^n$ is called \emph{recursive} if, for each $i \in \{2,\dots,n-1\}$, there exists an APN function $F_i \colon \ftwo^i \rightarrow \ftwo^i$ such that $F_2 \prec F_3 \prec \dots \prec F_{n-1} \prec F_n$.
\end{definition}

In Appendix~\ref{app:recursive}, we give an example of a recursive APN function in dimension $n=8$. Motivated by the new APN extensions we construct in Section~\ref{sec:r-extendable} below, we raise the following conjecture.

\begin{conjecture}
\label{con:recursive}
There exists a recursive APN function in every dimension $n \in \mathbb{N},n\geq 2$.
\end{conjecture}

\section{New APN Extensions in Dimension Eight}
\label{sec:r-extendable}
In this section, up to EA-equivalence, we construct new quadratic APN functions in dimension eight as extensions of quadratic 7-bit APN functions. Note that a search for quadratic APN functions $F \colon \ftwo^{n+1} \rightarrow \ftwo^{n+1}$ which contain a quadratic APN function $G \colon \ftwo^{n} \rightarrow \ftwo^{n}$ as a restriction was already conducted in~\cite{DBLP:journals/iacr/KalginI20a}. In this previous work, the search was based on necessary properties of the algebraic normal form of $F$ and it was quite similar to the QAM approach~\cite{DBLP:journals/dcc/YuWL14}. However, no results were reported for $n \geq 7$. 

The following proposition derives a simple form of such APN functions, which we utilize in our search.

\begin{proposition}
\label{prop:standard_form}
Let $F \colon \ftwo^{n+1} \rightarrow \ftwo^{n+1}$ be a quadratic function. Then, there exists a function $G\colon \ftwo^n \rightarrow \ftwo^n$ of algebraic degree at most 2, a Boolean function $r \colon \ftwo^n \rightarrow \ftwo^n$ of algebraic degree at most 2, and two linear functions $L \colon \ftwo^n \rightarrow \ftwo^n$, $\ell \colon \ftwo^n \rightarrow \ftwo$ such that $F$ is EA-equivalent to
\begin{eqnarray*}
T:\ftwo^n \times \ftwo &\to& \ftwo^n \times \ftwo \\
\left(\begin{array}{c}
     x  \\
     y 
\end{array}\right) &\mapsto& 
\left(\begin{array}{c}
     G(x)  \\
     r(x) 
\end{array}\right)+
\left(\begin{array}{c}
     L(x)  \\
     \ell(x) 
\end{array}\right)\cdot y\;.
\end{eqnarray*}
\end{proposition}
\begin{proof}
By applying an EA-transformation to $F$, we can obtain a function 
\begin{eqnarray*}
T:\ftwo^n \times \ftwo \to \ftwo^n \times \ftwo,
\left(\begin{array}{c}
     x  \\
     y 
\end{array}\right) \mapsto
\left(\begin{array}{c}
     H(x,y)  \\
     h(x,y) 
\end{array}\right)\;,
\end{eqnarray*}
where $H\colon \ftwo^n \times \ftwo \rightarrow \ftwo^n$ is of algebraic degree at most 2, $H(0,0)=H(0,1) = 0$ and $h \colon \ftwo^n \times \ftwo \to \ftwo$ is of algebraic degree at most 2 with $h(0,0)=h(0,1) = 0$. Then, by defining $G(x) = H(x,0)$ and $r(x) = h(x,0)$, we obtain 
\begin{align}\label{eq:different_repr_standard_form}T\left(\begin{array}{c}
     x  \\
     y 
\end{array}\right) = \left(\begin{array}{c}
     G(x)  \\
     r(x) 
\end{array}\right)(y+1) + \left(\begin{array}{c}
     H(x,1)  \\
     h(x,1) 
\end{array}\right)y\;.\end{align}
Note that since $(G,r)$ is a restriction of $(H,h)$ to a linear hyperplane, it is also of algebraic degree at most 2.
Now, $T$ is quadratic if and only if $x \mapsto (G(x)+H(x,1),r(x)+h(x,1))$ is of algebraic degree at most 1, which means that $H(x,1) = G(x) + L(x)$ for an affine function $L$ and $h(x,1) = r(x) + \ell(x)$ for an affine function $\ell$. Since we chose $(H(0,1),h(0,1)) = (H(0,0),h(0,0)) = (0,0)$, we have $(L(0),\ell(0))=(0,0)$ and both $L$ and $\ell$ must therefore be linear.
\end{proof}

\begin{definition}
\label{def:r-extendability}
Let $G \colon \ftwo^{n} \rightarrow \ftwo^{n}$ be a quadratic APN function and let $r \colon \ftwo^n \rightarrow \ftwo$ be a Boolean function of algebraic degree at most 2. The function $G$ is called $r$-extendable if there exist two linear functions $L \colon \ftwo^n \rightarrow \ftwo^n$, $\ell \colon \ftwo^n \rightarrow \ftwo$ such that 
\begin{eqnarray*}
T:\ftwo^n \times \ftwo &\to& \ftwo^n \times \ftwo \\
\left(\begin{array}{c}
     x  \\
     y 
\end{array}\right) &\mapsto& 
\left(\begin{array}{c}
     G(x)  \\
     r(x) 
\end{array}\right)+
\left(\begin{array}{c}
     L(x)  \\
     \ell(x) 
\end{array}\right)\cdot y
\end{eqnarray*}
is APN. If $T$ is APN, we say that the tuple $(G,r,L,\ell)$ \emph{yields an APN function} $T$ and we say that $T$ is an \emph{APN extension of $G$ in standard form}.
\end{definition}

\begin{remark}
For any linear mappings $L \colon \ftwo^n \rightarrow \ftwo^n, \ell \colon \ftwo^n \rightarrow \ftwo$ and any Boolean function $r \colon \ftwo^n \rightarrow \ftwo$, we remark that if two functions $G,G' \colon \ftwo^n \rightarrow \ftwo^n$ are EA-equivalent, there exist linear mappings  $L' \colon \ftwo^n \rightarrow \ftwo^n, \ell' \colon \ftwo^n \rightarrow \ftwo$ and a Boolean function $r' \colon \ftwo^n \rightarrow \ftwo$ EA-equivalent to $r$ such that the two functions
\begin{align*}
\left(\begin{array}{c}
     x  \\
     y 
\end{array}\right) \mapsto
\left(\begin{array}{c}
     G(x)  \\
     r(x) 
\end{array}\right)+
\left(\begin{array}{c}
     L(x)  \\
     \ell(x) 
\end{array}\right)\cdot y, \quad
\left(\begin{array}{c}
     x  \\
     y 
\end{array}\right) \mapsto 
\left(\begin{array}{c}
     G'(x)  \\
     r'(x) 
\end{array}\right)+
\left(\begin{array}{c}
     L'(x)  \\
     \ell'(x) 
\end{array}\right)\cdot y
\end{align*}
are EA-equivalent as well. Further, note that if $T$ is given in the same form as in Proposition~\ref{prop:standard_form}, the EA-equivalence class of $G$ is contained within the trim spectrum of $T$ by choosing $H = \ftwo^n \times \{0\}$ and $\beta = e_{n+1}$.
\end{remark}

Using a recursive tree search similar to the approach described in~\cite{DBLP:journals/tit/BeierleBL21,DBLP:journals/corr/abs-2009-07204}, we conducted a search for 8-bit quadratic APN functions that are an extension of a quadratic APN function in dimension seven. In particular, we applied the following search procedure:
\begin{enumerate}
    \item Fix a representative $G \colon \ftwo^7 \rightarrow \ftwo^7$ of one of the 488 EA-equivalence classes of 7-bit quadratic APN functions.
    \item Guess the Boolean function $r$ of a possible APN extension $T$ of $G$ in standard form, i.e., randomly choose a function $r \colon \ftwo^n \rightarrow \ftwo$ of algebraic degree at most 2.
    \item Recursively construct the linear function $(L,\ell) \colon \ftwo^n \times \ftwo \rightarrow \ftwo^n \times \ftwo$ such that $T$ is APN. In case there is a contradiction with the property of $T$ being APN, the recursive algorithm backtracks. If a suitable function $(L,\ell)$ cannot be found after a predetermined number of iterations, we abort the search. Otherwise, we found an APN extension of $G$.
\end{enumerate}

By repeatedly iterating the above algorithm, we found 6,368 new quadratic APN functions in dimension eight up to EA-equivalence. With our implementation, it takes about 2 CPU hours to find one APN extension in dimension eight. The new APN functions are available in the dataset~\cite{our_functions} and our implementation is available at~\cite{trims_code}. The check for EA-inequivalence to the previously-known APN functions was done using the method explained in~\cite{ortho_paper}. In particular, we computed the extended Walsh spectra and the differential spectra of the ortho-derivatives of all the known and found quadratic 8-bit APN functions, which is a strongly discriminating EA-invariant.

As our new 8-bit APN functions are extensions of quadratic 7-bit APN functions by construction, they are all connected to some nodes at height 7 in the trimming graph. From the previously-known 8-bit APN functions (see Figure~\ref{fig:trimming-graph-n8}), one might expect that there are only a few of the quadratic 7-bit APN functions that can be extended to quadratic 8-bit APN functions. However, it turns out that the vast majority of the quadratic APN functions in dimension $n=7$ can be extended to a quadratic APN function in dimension $n=8$. In fact, this observation is the reason we raised Conjecture~\ref{con:recursive}.

\begin{remark}
    It is possible to further restrict the definition of an ``APN extension of $G$  in standard form''. Let $\mathcal{Q}_n$ be the set of quadratic homogeneous Boolean functions mapping $n$ bits to 1, i.e. the set of $n$-bit quadratic Boolean functions with no linear or constant terms. Furthermore, let $G$ be an $n$-bit quadratic APN function, and let $(G,r,L,\ell)$ yield an APN function $T$. Then we can impose for $r$ to be in $\mathcal{Q}_n$, and for all coordinates $G_i$ of $G$ to be in $\mathcal{Q}_n$ as well: removing all affine terms is the same as adding an affine function to the output of $T$, which would not change its EA-equivalence class. This first restriction means that $r$ can be searched for in a space of dimension $n(n-1)/2$.
    We can go further. In the procedure outlined above for extending a given function $G$, we first guess $r$ and then find $(L, \ell)$. For a given $r$, setting $r' = r + \sum_i \epsilon_i G_i$ for some $\epsilon \in \ftwo^n$ and then searching for all $(L, \ell)$ would yield the same EA-equivalence classes as starting from $r$. Indeed, we would simply find pairs $(L, \ell')$ where $\ell$ is replaced by $\ell' = \ell + \sum_i \epsilon_i L_i$. Thus, we can safely search for $r$ in the complement of the span of the coordinates of $G$, i.e. in a space of dimension $n(n-1)/2 - n$. For $n=7$, this quantity is only equal to 15.
    Unfortunately, for a given tuple $(G\colon \ftwo^7 \rightarrow \ftwo^7,r\colon \ftwo^7 \rightarrow \ftwo)$, exhausting all pairs $(L,\ell)$ using our implementation already takes several minutes, and exhausting all possible such pairs for each of the $488 \cdot 2^{15}$ choices of $(G,r)$ is simply infeasible. Thus, while this reduction of the search space would work, it would be of no practical impact at this stage due to the cost of finding $(L, \ell)$. Still, should progress be made in this direction, it might become possible to exhaustively find all quadratic 8-bit APN extensions of quadratic APN functions in dimension $n=7$.
\end{remark}

\paragraph{The Number of Known Instances of 8-Bit APN Functions} The 6,368 APN functions constructed in this work together with the 23 APN functions listed in~\cite{DBLP:journals/amco/EdelP09} (including one non-quadratic monomial function), the 8,157 APN functions constructed by the QAM method~\cite{DBLP:journals/iacr/YuWL13}, the 10 APN functions presented in~\cite{weng2013quadratic}, the two APN functions from the Taniguchi family~\cite{DBLP:journals/dcc/Taniguchi19a}, the 12,921 APN functions from~\cite{DBLP:journals/corr/abs-2009-07204}, and the 5,412 APN functions presented in~\cite{DBLP:journals/iacr/YuP21}, are all APN functions in dimension eight up to CCZ-equivalence known until January 2022 (32,893 in total).

\section{Quadratic APN Functions With Maximum Linearity}
\label{sec:0-extendable}
The following simple observation illustrates how a quadratic function in dimension $n+1$ and having linearity $2^n$ can be obtained as an extension of a function on $n$ bit. A similar argument (for APN functions) was already given in~\cite[Remark 12]{DBLP:journals/tit/Carlet18}.
\begin{proposition}
\label{prop:high_linearity}
Let $F \colon \ftwo^{n+1} \rightarrow \ftwo^{n+1}$ be a quadratic function with linearity $2^n$. Then, there exists a function $G\colon \ftwo^n \rightarrow \ftwo^n$ of algebraic degree at most 2 and two linear functions $L \colon \ftwo^n \rightarrow \ftwo^n$, $\ell \colon \ftwo^n \rightarrow \ftwo, \ell \neq 0$ such that $F$ is EA-equivalent to
\begin{eqnarray*}
T:\ftwo^n \times \ftwo &\to& \ftwo^n \times \ftwo \\
\left(\begin{array}{c}
     x  \\
     y 
\end{array}\right) &\mapsto& 
\left(\begin{array}{c}
     G(x)  \\
     0 
\end{array}\right)+
\left(\begin{array}{c}
     L(x)  \\
     \ell(x) 
\end{array}\right)\cdot y\;.
\end{eqnarray*}
\end{proposition}
\begin{proof}
Since $F \colon \ftwo^{n+1} \rightarrow \ftwo^{n+1}$ is a quadratic function with linearity $2^n$, it consists of a component which is EA-equivalent to $(x,y) \mapsto \ell(x)y$ for a linear function $\ell \colon \ftwo^n \rightarrow \ftwo, \ell \neq 0$, see~\cite[Thm.\@10 and p.\@ 173]{carlet_2021}.
Therefore, by applying an EA-transformation to $F$, we can obtain a function 
\begin{eqnarray*}
T:\ftwo^n \times \ftwo \to \ftwo^n \times \ftwo,
\left(\begin{array}{c}
     x  \\
     y 
\end{array}\right) \mapsto
\left(\begin{array}{c}
     H(x,y)  \\
     \ell(x)y 
\end{array}\right)\;,
\end{eqnarray*}
where $H\colon \ftwo^n \times \ftwo \rightarrow \ftwo^n$ is of algebraic degree at most 2, $H(0,0)=H(0,1) = 0$ and $\ell \colon \ftwo^n \to \ftwo, \ell \neq 0$ is linear. Then, by defining $G(x) = H(x,0)$, we obtain 
\begin{align}\label{eq:different_repr}T\left(\begin{array}{c}
     x  \\
     y 
\end{array}\right) = \left(\begin{array}{c}
     G(x)  \\
     0 
\end{array}\right)(y+1) + \left(\begin{array}{c}
     H(x,1)  \\
     \ell(x) 
\end{array}\right)y\;.\end{align}
Note that since $G$ is a restriction of $H$ to a linear hyperplane, it is also of algebraic degree at most 2.
Now, $T$ is quadratic if and only if $x \mapsto G(x)+H(x,1)$ is of algebraic degree at most 1, which means that $H(x,1) = G(x) + L(x)$ for an affine function $L$. Since we chose $H(0,1) = H(0,0) = 0$, we have $L(0)=0$ and $L$ must therefore be linear.
\end{proof}

\begin{remark}
\label{rem:extendability_ea}
For any linear mappings $L \colon \ftwo^n \rightarrow \ftwo^n, \ell \colon \ftwo^n \rightarrow \ftwo, \ell \neq 0$, we remark that if two functions $G,G' \colon \ftwo^n \rightarrow \ftwo^n$ are EA-equivalent, there exist linear mappings  $L' \colon \ftwo^n \rightarrow \ftwo^n, \ell' \colon \ftwo^n \rightarrow \ftwo, \ell' \neq 0$ such that the two functions
\begin{align*}
\left(\begin{array}{c}
     x  \\
     y 
\end{array}\right) \mapsto
\left(\begin{array}{c}
     G(x)  \\
     0 
\end{array}\right)+
\left(\begin{array}{c}
     L(x)  \\
     \ell(x) 
\end{array}\right)\cdot y, \quad
\left(\begin{array}{c}
     x  \\
     y 
\end{array}\right) \mapsto 
\left(\begin{array}{c}
     G'(x)  \\
     0 
\end{array}\right)+
\left(\begin{array}{c}
     L'(x)  \\
     \ell'(x) 
\end{array}\right)\cdot y
\end{align*}
are EA-equivalent as well. Thus, the function $G$ as in Proposition~\ref{prop:high_linearity} can be chosen up to EA-equivalence.
\end{remark}

 We now deduce when the functions $T$ of the same form as given in Proposition~\ref{prop:high_linearity} are APN.

\begin{theorem}
\label{thm:high_linearity_apn}
Let $G:\ftwo^n \to \ftwo^n$ be a function of algebraic degree at most 2, $L:\ftwo^n \to \ftwo^n$ be linear and $\ell:\ftwo^n \to \ftwo, \ell \neq 0$ be linear. Then
\begin{eqnarray*}
T:\ftwo^n \times \ftwo &\to& \ftwo^n \times \ftwo \\
\left(\begin{array}{c}
     x  \\
     y 
\end{array}\right) &\mapsto& 
\left(\begin{array}{c}
     G(x)  \\
     0 
\end{array}\right)+
\left(\begin{array}{c}
     L(x)  \\
     \ell(x) 
\end{array}\right)\cdot y
\end{eqnarray*}
is APN if and only if the following two assertions hold:
\begin{enumerate}
\item $G$ is APN.
\item $\langle \pi_G(\alpha), L(\alpha) \rangle=1$ for all $\alpha \in \ftwo^n \setminus \{0\}$ with $\ell(\alpha)=0$.
\end{enumerate}
\end{theorem}
\begin{proof}
By definition, the function $T$ is APN if and only if, for all $\alpha, \gamma \in \ftwo^n$, $\beta, \delta \in \ftwo$ with $(\alpha,\beta) \neq (0,0)$, the equation 
\begin{equation}
\label{eq:apn_T}
    T\left( \begin{array}{c}
     x  \\
     y 
\end{array}\right) + T\left( \begin{array}{c}
     x+\alpha  \\
     y + \beta 
\end{array}\right) = \left( \begin{array}{c}
    \gamma  \\
     \delta 
\end{array} \right)
\end{equation}
has at most two solutions $(x,y) \in \ftwo^n \times \ftwo$. By the definition of $T$, Equation~(\ref{eq:apn_T}) is equivalent to
\begin{equation}
\label{eq:apn_T2}
    \left( \begin{array}{c}
     G(x) + G(x + \alpha)  \\
     0 
\end{array}\right) + \left( \begin{array}{c}
     L(\alpha) \\
     \ell(\alpha) 
\end{array}\right)y + \left( \begin{array}{c}
     L(x) \\
     \ell(x) 
\end{array}\right)\beta + \left( \begin{array}{c}
     L(\alpha) \\
     \ell(\alpha) 
\end{array}\right)\beta = \left( \begin{array}{c}
    \gamma  \\
     \delta 
\end{array} \right)\;.
\end{equation}

For every $\alpha \in \ftwo^n$, we recall that $B_{\alpha}$ is defined as $B_{\alpha} \colon \ftwo^n \rightarrow \ftwo^n, x \mapsto G(x) + G(x+\alpha) + G(\alpha) + G(0)$, which is a linear mapping if $G$ is of algebraic degree at most 2. We now consider two cases.

\paragraph{Case $\beta = 0$, $\alpha \neq 0$} For $y=0$, Equation~(\ref{eq:apn_T2}) simplifies to the system $\left(B_{\alpha}(x) = \gamma + G(\alpha) + G(0) \right)\wedge (\delta = 0)$, which directly yields that $G$ being APN is a necessary condition for $T$ being APN. Conversely, if $G$ is APN, this system has at most 2 solutions $x \in \ftwo^n$.

For $y = 1$, Equation~(\ref{eq:apn_T2}) simplifies to
\begin{eqnarray}
\label{eq:apn_T3}
    \left( \begin{array}{c}
     B_{\alpha}(x) \\
     \ell(\alpha) 
\end{array}\right) = \left( \begin{array}{c}
     \gamma + G(\alpha) + G(0) + L(\alpha) \\
     \delta
\end{array}\right)\;.
\end{eqnarray}

Similarly as above, Equation~(\ref{eq:apn_T3}) has no more than two solutions $x \in \ftwo^n$ if $G$ is APN. For $T$ being APN, we also need to require that Equation~(\ref{eq:apn_T3}) has no solutions $x \in \ftwo^n$ in the case where both $\ell(\alpha) = 0$ and $\gamma + G(\alpha) + G(0) \in \img{B_{\alpha}}$. This is equivalent to the condition that $L(\alpha) \notin \img{B_{\alpha}}$ whenever $\ell(\alpha) = 0$.
Since $\img{B_{\alpha}} = \{ x \in \ftwo^n \mid \langle \pi_G(\alpha),x \rangle = 0\}$ we obtain that $\langle \pi_G(\alpha),L(\alpha) \rangle = 1$ for all $\alpha \in \ftwo^n \setminus \{0\}$ with $\ell(\alpha ) = 0$.

To summarize, the two conditions 1 and 2 in the statement of the theorem are necessary and sufficient for Equation~(\ref{eq:apn_T2}) having at most 2 solutions $(x,y) \in \ftwo^n \times \ftwo$ for all $\alpha,\gamma \in \ftwo^n$, $\beta,\delta \in \ftwo$ with $\beta = 0, \alpha \neq 0$.

\paragraph{Case $\beta = 1$} We will show that Condition 2 of the statement is a sufficient condition for Equation~(\ref{eq:apn_T2}) having at most 2 solutions $(x,y) \in \ftwo^n \times \ftwo$. 

For $y=1$, Equation~(\ref{eq:apn_T2}) simplifies to 
\begin{eqnarray}
\label{eq:apn_T5}
    \left( \begin{array}{c}
     G(x) + G(x+\alpha) + L(x) \\
     \ell(x) 
\end{array}\right) = \left( \begin{array}{c}
     \gamma \\
     \delta
\end{array}\right)
\end{eqnarray}
and for $y=0$, Equation~(\ref{eq:apn_T2}) simplifies to
\begin{eqnarray}
\label{eq:apn_T6}
    \left( \begin{array}{c}
     G(x) + G(x+\alpha) + L(x+\alpha) \\
     \ell(x+\alpha) 
\end{array}\right) = \left( \begin{array}{c}
     \gamma \\
     \delta
\end{array}\right)\;.
\end{eqnarray}
An element $x \in \ftwo^n$ is a solution of Equation~(\ref{eq:apn_T5}) if and only if $x+\alpha$ is a solution of Equation~(\ref{eq:apn_T6}). Therefore, for Equation~(\ref{eq:apn_T2}) having at most two solutions $(x,y) \in \ftwo^n \times \ftwo$, we need that Equation~(\ref{eq:apn_T5}) has at most 1 solution $x \in \ftwo^n$. In other words, we need to show that the mapping
\begin{eqnarray*}
    H_{\alpha}\colon \ftwo^n \rightarrow \ftwo^n \times \ftwo, x \mapsto \left( \begin{array}{c}
     G(x) + G(x+\alpha) + L(x) \\
     \ell(x) 
\end{array}\right)
\end{eqnarray*}
is injective. For this, let us choose an arbitrary element $x \in \ftwo^n$ and a non-zero element $w \in \ftwo^n$ and consider $H_{\alpha}(x) + H_{\alpha}(x +w)$, which is equal to
\begin{eqnarray}
\label{eq:apn_T7}
\left( \begin{array}{c}
     G(x) + G(x+\alpha) + G(x+w) + G(x+w+\alpha) + L(w) \\
     \ell(w) 
\end{array}\right)\;.
\end{eqnarray}
Since $G$ is of algebraic degree at most 2, the mapping $x \mapsto G(x) + G(x+\alpha) + G(x+w) + G(x+w+\alpha)$ is constant. Thus, 
\begin{eqnarray*}
H_{\alpha}(x) + H_{\alpha}(x+w) = H_{\alpha}(0) + H_{\alpha}(w) = \left( \begin{array}{c}
     B_w(\alpha) + L(w) \\
     \ell(w) 
\end{array}\right)\;.
\end{eqnarray*}
The right-hand side can only be equal to 0 if $L(w) \in \img{B_w}$ and $\ell(w) = 0$. Therefore, if $L(w) \notin \img{B_w}$ for all $w \in \ftwo^n \setminus \{0\}$ with $\ell(w) = 0$, the mapping $H_{\alpha}$ must be injective.
\end{proof}

Note that Condition 2 of Theorem~\ref{thm:high_linearity_apn} implies that the mapping \begin{eqnarray*}
H_0 \colon \ftwo^n &\rightarrow \ftwo^n \times \ftwo \\
x &\mapsto
\left( \begin{array}{c}
     L(x) \\
     \ell(x) 
\end{array}\right)
\end{eqnarray*}
has full rank equal to $n$. In particular, if $T$ is APN, $L$ can only be of rank $n$ or $n-1$. We further remark that, if $n>2$ and $T$ is APN, the APN function $G$ must actually be quadratic (since affine functions in dimension $n>2$ cannot be APN).

\begin{corollary}
Let $G \colon \ftwo^n \rightarrow \ftwo^n$ be a quadratic APN function. $G$ is 0-extendable if there exist linear functions $L:\ftwo^n \to \ftwo^n$ and $\ell:\ftwo^n \to \ftwo, \ell \neq 0$ such that $\langle \pi_G(x), L(x) \rangle=1$ for all $x \in \ftwo^n \setminus \{0\}$ with $\ell(x)=0$.
\end{corollary}

From Remark~\ref{rem:extendability_ea}, it is obvious that the property of being 0-extendable is invariant under EA-equivalence. Moreover, since quadratic APN functions in odd dimension must be almost bent, 0-extendable APN functions can only exist in dimension $n=2$ or in odd dimension $n$. 

Let us now define the set 
\begin{align*}\Gamma_{G,\ell} \hspace{-.2em}\coloneqq  \hspace{-.2em} \{L \in \mathcal{L}(\ftwo^n,\ftwo^n) \mid  \langle \pi_G(x),L(x) \rangle = 1 \text{ for all } x \in \ftwo^n \setminus \{0\} \text{ with } \ell(x) = 0\},\end{align*}
which is either empty or an affine subspace of $\mathcal{L}(\ftwo^n,\ftwo^n)$. Given $G$ and $\ell$, the set $\Gamma_{G,\ell}$ can be recovered by first computing the ortho-derivative of $G$ and by then solving a system of $2^{n-1}$ linear equations with $n^2$ unknowns. The following proposition yields a lower bound on $\lvert \Gamma_{G,\ell} \rvert$ in cases where $\Gamma_{G,\ell} \neq \emptyset$.

\begin{proposition}
\label{prop:many_L}
Let $n \in \mathbb{N}, n\geq 3$. Let $G:\ftwo^n \to \ftwo^n$ be a quadratic mapping, $L:\ftwo^n \to \ftwo^n$ be linear and $\ell:\ftwo^n \to \ftwo, \ell \neq 0$ be linear such that $(G,0,L,\ell)$ yields an APN function $T$. For all $\mu, \nu \in \ftwo^n$, the tuple $(G,0,L+B_{\mu}+\ell \nu,\ell)$ yields an APN function EA-equivalent to $T$, where $B_{\mu}\colon \ftwo^n \rightarrow \ftwo^n, x \mapsto G(x) + G(x+\mu) + G(\mu) + G(0)$.

The functions $L + B_{\mu} + \ell \nu$ are pairwise distinct for $\mu,\nu \in \ftwo^n$. Moreover, for every $\mu \in \ftwo^n$, we have 
\begin{align*} 2^{n-1} &= \lvert\left\{\nu \in \ftwo^n \mid \rank(L + B_{\mu} + \ell \nu) = n\right\} \rvert \\
&= \lvert\left\{\nu \in \ftwo^n \mid \rank(L + B_{\mu} + \ell \nu) = n-1\right\}\rvert\;. \end{align*}
\end{proposition}
\begin{proof}
For an element $c \in \ftwo^n$, let us consider the linear involution
\begin{eqnarray*}
M_{c}:\ftwo^n \times \ftwo &\to& \ftwo^n \times \ftwo \\
\left(\begin{array}{c}
     x  \\
     y 
\end{array}\right) &\mapsto& 
\left(\begin{array}{c}
     x + c y  \\
     y
\end{array}\right)\;.\end{eqnarray*}

For $\mu \in \ftwo^n$, let us consider the function
\begin{eqnarray*}
T'_{\mu}:\ftwo^n \times \ftwo &\to& \ftwo^n \times \ftwo \\
\left(\begin{array}{c}
     x  \\
     y 
\end{array}\right) &\mapsto& 
TM_{\mu}\left(\begin{array}{c}
     x  \\
     y
\end{array}\right)+
\left(\begin{array}{c}
     L(\mu) + G(\mu) + G(0)  \\
     \ell(\mu) 
\end{array}\right)\cdot y\;.
\end{eqnarray*}
By definition, $T'_{\mu}$ is EA-equivalent to $T$. Using the representation of $T$ given in Equation~(\ref{eq:different_repr}), we obtain
\begin{align*}
   T'_{\mu}\left(\begin{array}{c}
     x  \\
     y 
\end{array}\right) &= \left(\begin{array}{c}
     G(x+\mu y)  \\
     0 
\end{array}\right)(y+1) + \left(\begin{array}{c}
     G(x+\mu y) + L(x+\mu y)  \\
     \ell(x+ \mu y) 
\end{array}\right)y \\ 
&+ \left(\begin{array}{c}
     L(\mu) + G(\mu) + G(0)  \\
     \ell(\mu) 
\end{array}\right)y \\
 &= \left(\begin{array}{c}
     G(x)  \\
     0 
\end{array}\right)(y+1) + \left(\begin{array}{c}
     G(x+\mu) + G(\mu) + G(0) + L(x)  \\
     \ell(x) 
\end{array}\right)y \\
&= \left(\begin{array}{c}
     G(x)  \\
     0 
\end{array}\right)(y+1) + \left(\begin{array}{c}
     G(x) + L(x) + B_{\mu}(x)  \\
     \ell(x)
\end{array}\right)y \\
&= \left(\begin{array}{c}
     G(x)  \\
     0 
\end{array}\right) + \left(\begin{array}{c}
     (L + B_{\mu})(x)  \\
     \ell(x) 
\end{array}\right)y\;,
\end{align*}
which is of the same form as in Theorem~\ref{thm:high_linearity_apn}. Thus, $(G,0,L+B_{\mu},\ell)$ yields the APN function $T'_{\mu}$.

For an element $\nu \in \ftwo^n$, we obtain 
\begin{align*}
   M_{\nu} T\left(\begin{array}{c}
     x  \\
     y 
\end{array}\right) = \left(\begin{array}{c}
     G(x)  \\
     0 
\end{array}\right) + \left(\begin{array}{c}
     (L + \ell \nu)(x)  \\
     \ell(x) 
\end{array}\right)y\;.
\end{align*}
The function $M_{\nu}T$ is also EA-equivalent to $T$ and of the same form as in Theorem~\ref{thm:high_linearity_apn}. Thus, $(G,0,L+\ell \nu,\ell)$ yields the APN function $M_{\nu}T$. Combining the above, we obtain that $(G,0,L+B_{\mu}+\nu \ell,\ell)$ yields the APN function $M_{\nu}T'_{\mu}$, which is EA-equivalent to $T$.

To see that the functions $L+ B_{\mu} + \ell \nu$ are pairwise distinct for all $\mu,\nu \in \ftwo^n$, we need to show that the linear mapping 
\begin{eqnarray*}
J:\ftwo^n \times \ftwo^n &\to& \mathcal{L}(\ftwo^n,\ftwo^n) \\
\left(\begin{array}{c}
     \mu  \\
     \nu 
\end{array}\right) &\mapsto& 
B_{\mu} + \ell \nu \end{eqnarray*}
is injective. We show that the kernel of $J$ is trivial. Suppose that, for all $x \in \ftwo^n$, we have $B_{\mu}(x) + \ell(x)\nu = 0$. This implies $B_{\mu}(x) = \nu$ for all $x \in \ftwo^n$ with $\ell(x) = 1$. Since $\ell$ is not the zero mapping and since $B_{\mu}$ is 2-to-1 for every non-zero $\mu$, we have $\mu = 0$. From $\mu=0$, we immediately deduce $\nu = 0$, so the kernel of $J$ is trivial. 

To prove the last statement, let us fix an element $\mu \in \ftwo^n$. We consider $L' \coloneqq L + B_{\mu}$. Suppose that $\rank(L') \neq n$. We can choose an element $\nu \in \ftwo^n$ such that $\nu \notin \img{L'}$. For such $\nu$, we have that $\rank(L' + \ell \nu) = n$. This can be observed by looking at the kernel of $L' + \ell \nu$. For all $x \in \ftwo^n$ with $\ell(x) = 0$, we have $L'(x) + \ell(x) \nu = L'(x) = 0$ if and only if $x=0$, since $\rank((L',\ell)) = n$. For all $x \in \ftwo^n$ with $\ell(x) = 1$, we have $L'(x) +  \ell(x) \nu = 0$ if and only if $L'(x) = \nu$, which is not possible since $\nu \notin \img{L'}$. Therefore, there exists an element $\nu \in \ftwo^n$ such that $L'' \coloneqq L + J(\mu,\nu)$ is invertible. By using a similar argument as above we obtain that, for all $\nu' \in \ftwo^n$, the mapping $L'' +  \ell \nu'$ is invertible if and only if $\ell(L''^{-1}(\nu')) = 0$. The statement follows since there are  exactly $2^{n-1}$ elements $\nu' \in \ftwo^n$ with $\ell(L''^{-1}(\nu')) = 0$.
\end{proof}

Let $G \colon \ftwo^n \rightarrow \ftwo^n$ be a quadratic APN function and $\ell \colon \ftwo^n \rightarrow \ftwo, \ell \neq 0$ be linear. The above proposition states that, if $\Gamma_{G,\ell} \neq \emptyset$, we have $\lvert \Gamma_{G,\ell}\rvert \geq 2^{2n}$. We say that two elements $L,L' \in \Gamma_{G,\ell}$ are \emph{$\Gamma$-equivalent} if there exist $\mu,\nu \in \ftwo^n$ such that $L' = L + B_{\mu} + \ell \nu$. As it was shown in Proposition~\ref{prop:many_L}, if $L,L' \in \Gamma_{G,\ell}$ are $\Gamma$-equivalent, the tuples $(G,0,L,\ell)$ and $(G,0,L',\ell)$ yield EA-equivalent APN functions.

\subsection{A Classification in Dimension Eight}
Recently, Kalgin and Idrisova completely classified all quadratic 7-bit APN functions~\cite{DBLP:journals/iacr/KalginI20a}. In total, there are 488 such functions up to EA-equivalence. By using this classification and by applying the results from Theorem~\ref{thm:high_linearity_apn} and Proposition~\ref{prop:many_L}, we can now classify \emph{all} 8-bit quadratic APN functions with linearity $2^7$.

Let $\mathcal{G} = \{G_1,G_2,\dots,G_{488}\}$ be a set of 7-bit quadratic APN functions that contains one representative of each EA-equivalence class. Algorithm~\ref{alg:classification} describes how $\mathcal{G}$ can be used to obtain a classification of 8-bit quadratic APN functions with maximum linearity up to EA-equivalence.

\begin{algorithm}
\caption{Classification of 8-bit quadratic APN functions with maximum linearity up to EA-equivalence}\label{alg:classification}
\begin{algorithmic}[1]
    \Function{ZeroExtensions}{$G$} \Comment{$G \colon \ftwo^n \rightarrow \ftwo^n$ is a quadratic APN function}
    \State $\mathrm{sol} \gets \emptyset$
    \For{$\gamma \in \ftwo^n \setminus \{0\}$}
	    \State Compute $\Gamma_{G,x \mapsto \langle \gamma,x \rangle}$ by solving a linear system
	    \For{$L \in \Gamma_{G,x \mapsto \langle \gamma,x \rangle}$ up to $\Gamma$-equivalence}
	    \State $\mathrm{sol} \gets \mathrm{sol} \cup T$,  where $T$ is such that $(G,0,L,x \mapsto \langle \gamma,x \rangle)$ yields $T$
	    \EndFor
	    \EndFor
	    \State \Return $\mathrm{sol}$
    \EndFunction
    \vspace{1em}
    \State $\mathrm{sol} \gets \emptyset$
	\For{$G \in \mathcal{G}$}
	    \State $\mathrm{sol} \gets \mathrm{sol} \cup \textsc{ZeroExtensions}(G)$
	\EndFor
	\State \Return $\mathrm{sol}$
\end{algorithmic}
\end{algorithm}

A sage~\cite{sagemath} implementation of Algorithm~\ref{alg:classification} is given in Appendix~\ref{app:sage}. For each quadratic 7-bit APN function $G$, the running time of the $\textsc{ZeroExtensions}$ procedure is less than one second on a PC. To execute this code, the library sboxU~\cite{Perrin_sboxu} is needed.

There are exactly four functions $G$ in $\mathcal{G}$ for which there exists a $\gamma \in \ftwo^7$ such that that $\Gamma_{G,x \mapsto \langle \gamma,x \rangle}$ is not empty. Moreover, for each of those four functions, there is exactly one such $\gamma$. The  space $\Gamma_{G,x \mapsto \langle \gamma,x \rangle}$ is of size $2^{14}$ in all those cases, implying that there is only one element in $\Gamma_{G,x \mapsto \langle \gamma,x \rangle}$ up to $\Gamma$-equivalence. Thus, Algorithm~\ref{alg:classification} outputs four 8-bit quadratic APN functions with linearity $2^7$. Those are pairwise inequivalent up to EA-equivalence.  More precisely, they correspond to the four EA-equivalence classes reported in~\cite{DBLP:journals/corr/abs-2009-07204}. While in~\cite{DBLP:journals/corr/abs-2009-07204}, the authors have shown that there are \emph{at least} four pairwise EA-inequivalent quadratic APN functions in dimension eight with maximum linearity, by using Algorithm~\ref{alg:classification} we now deduced that there are \emph{exactly} four such EA-equivalence classes. 
\begin{theorem}
In dimension eight, there are exactly four quadratic APN functions with linearity $2^7$ up to EA-equivalence.
\end{theorem}

We give representatives of those four EA-equivalence classes in Section~\ref{sec:simpler} using a simpler representation.

\paragraph{Searching for 0-Extensions of 9-Bit and 11-Bit Quadratic APN Functions}  For all of the 60 known instances of 9-bit quadratic APN functions, we checked whether there exists a $\gamma \in \ftwo^9$ such that $\Gamma_{G,x \mapsto \langle \gamma,x \rangle}$ is non-empty. This is never the case. Thus, up to now, we do not know a quadratic 9-bit APN function that can be extended to a quadratic 10-bit APN function with linearity $2^9$. 

We also checked whether any of the known quadratic 11-bit APN functions coming from a known infinite family of APN functions (see the list in~\cite{apn-list-11}) is 0-extendable. In particular, we checked the quadratic monomial functions and the function $\mathbb{F}_{2^{11}} \rightarrow \mathbb{F}_{2^{11}}, x \mapsto x^3 + \tr{x^9}$ discovered in~\cite{DBLP:journals/ffa/BudaghyanCL09}. Also, none of those functions is 0-extendable.

\subsection{A Simpler Representation}
\label{sec:simpler}
As a summary of previous results, we can assume without loss of generality that quadratic APN functions with maximum linearity are of the following form.
\begin{theorem}
\label{thm:canonical_form}
Let $n \in \mathbb{N},n\geq 3$ and let $\gamma \in \ftwo^n \setminus \{0\}$ be an arbitrary non-zero element. Let $T \colon \ftwo^n \times \ftwo \rightarrow \ftwo^n \times \ftwo$ be a quadratic APN function with linearity $2^n$. Then, there exist a quadratic APN function $G \colon \ftwo^n \rightarrow \ftwo^n$ such that $\langle \pi_G(\alpha),\alpha \rangle = 1$ for all $\alpha \in \ftwo^n \setminus \{0\}$ with $\langle \gamma,\alpha \rangle = 0$. More precisely, $T$ is EA-equivalent to
\begin{eqnarray*}
\ftwo^n \times \ftwo &\to& \ftwo^n \times \ftwo \\
\left(\begin{array}{c}
     x  \\
     y 
\end{array}\right) &\mapsto& 
\left(\begin{array}{c}
     G(x)  \\
     0 
\end{array}\right)+
\left(\begin{array}{c}
     x  \\
     \langle \gamma, x \rangle
\end{array}\right)\cdot y\;.
\end{eqnarray*}
\end{theorem}
\begin{proof}
By Theorem~\ref{thm:high_linearity_apn}, we know that there exist a quadratic APN function $Q \colon \ftwo^n \rightarrow \ftwo^n$, a linear mapping $L \colon \ftwo^n \rightarrow \ftwo^n$, and an element $\delta \in \ftwo^n \setminus \{0\}$ such that $(Q,0,L,x \mapsto \langle \delta,x \rangle)$ yields $T$. By Proposition~\ref{prop:many_L}, we can assume without loss of generality that $L$ is invertible (then $(Q,0,L,x \mapsto \langle \delta,x \rangle)$ does not necessarily yield $T$, but an APN function $T'$ EA-equivalent to $T$). We then have
\begin{eqnarray*}
T' \colon \ftwo^n \times \ftwo &\to& \ftwo^n \times \ftwo \\
\left(\begin{array}{c}
     x  \\
     y 
\end{array}\right) &\mapsto& 
\left(\begin{array}{c}
     Q(x)  \\
     0 
\end{array}\right)+
\left(\begin{array}{c}
     L(x)  \\
     \langle \delta, x \rangle
\end{array}\right)\cdot y\;.
\end{eqnarray*}
Let us choose an invertible linear mapping $L' \colon \ftwo^n \rightarrow \ftwo^n$ with $L'^{\top}(\delta) = \gamma$. It is straightforward to deduce that $T'$ is EA-equivalent to
\begin{eqnarray*}
\ftwo^n \times \ftwo &\to& \ftwo^n \times \ftwo \\
\left(\begin{array}{c}
     x  \\
     y 
\end{array}\right) &\mapsto& 
\left(\begin{array}{c}
     L'^{-1}L^{-1} \circ Q \circ L'(x)  \\
     0 
\end{array}\right)+
\left(\begin{array}{c}
     x  \\
     \langle \gamma, x \rangle
\end{array}\right)\cdot y\;.
\end{eqnarray*}
The result follows by defining $G \coloneqq (LL')^{-1} \circ Q \circ L'$.
\end{proof}

\begin{remark}
Let $g \in \F_{2^7}$ be an element with minimal polynomial $X^7 + X +1$. We can represent the four pairwise EA-inequivalent quadratic APN functions in dimension eight with maximum linearity by the four functions $T_8^{(1)},T_8^{(2)},T_8^{(3)},T_8^{(4)}$, where 
\begin{align*}
    T_8^{(i)} \colon \mathbb{F}_{2^7} \times \ftwo &\rightarrow \mathbb{F}_{2^7} \times \ftwo\\
    \left(\begin{array}{c}
     x  \\
     y 
\end{array}\right) &\mapsto 
\left(\begin{array}{c}
     G_i(x)  \\
     0 
\end{array}\right)+
\left(\begin{array}{c}
     x  \\
     \tr{x} 
\end{array}\right)\cdot y
\end{align*}
and where $G_i$ are the four 7-bit quadratic APN functions defined in univariate representation by
\begin{align*}
    G_1(x) &= g^{92}x^{96} + g^{50}x^{80} + g^{27}x^{72} + g^{28}x^{68} + x^{66} + g^{97}x^{65} + g^{60}x^{48} \\ &+ g^{88}x^{40} + g^{123}x^{36} + g^{43}x^{34} + g^{32}x^{33} + g^{26}x^{24} + g^{100}x^{20} + g^{115}x^{18} \\ &+ g^{85}x^{17} + g^{111}x^{12} + g^{28}x^{10} + g^{93}x^{9} + g^{113}x^{6} + g^{53}x^{5} + g^{10}x^{3}\\
    G_2(x) &= g^{68}x^{96} + g^{3}x^{80} + g^{58}x^{72} + g^{39}x^{68} + g^{43}x^{66} + g^{96}x^{65} + g^{118}x^{48} \\ &+ g^{102}x^{40} + g^{61}x^{36} + g^{69}x^{34} + g^{59}x^{33} + g^{110}x^{24} + g^{99}x^{20} + g^{53}x^{18} \\ &+ g^{63}x^{17} + g^{55}x^{12} + g^{98}x^{10} + g^{31}x^{9} + g^{57}x^{6} + g^{69}x^{5} + g^{87}x^{3} \\
    G_3(x) &= g^{71}x^{96} + g^{46}x^{80} + g^{15}x^{72} + g^{126}x^{68} + g^{44}x^{65} + g^{38}x^{48} \\ &+ g^{104}x^{40} + x^{36} + g^{73}x^{34} + g^{83}x^{33} + g^{38}x^{24} + g^{3}x^{20} + g^{120}x^{18} \\ &+ g^{34}x^{17} + g^{78}x^{12} + g^{108}x^{10} + g^{28}x^{9} + g^{113}x^{6} + g^{100}x^{5} + g^{70}x^{3}\\
    G_4(x) &= g^{71}x^{96} + g^{20}x^{80} + g^{125}x^{72} + g^{40}x^{68} + g^{71}x^{66} + g^{75}x^{65} + g^{113}x^{48} \\ &+ g^{100}x^{40} + g^{29}x^{36} + g^{62}x^{34} + g^{40}x^{33} + g^{97}x^{24} + g^{22}x^{20} + g^{111}x^{18} \\ &+ g^{106}x^{17} + g^{86}x^{12} + g^{29}x^{10} + gx^{9} + g^{64}x^{6} + g^{51}x^{5} + g^{16}x^{3}\;.
\end{align*}
\end{remark}

\paragraph{The Walsh Transform of Quadratic APN Functions with Maximum Linearity}
To deduce the Walsh spectrum of a quadratic APN function with maximum linearity, we need the following well-known result on the sum of fourth powers of the Walsh coefficients of APN functions. 

\begin{lemma}
\label{lem:4th_moments}{\cite{DBLP:conf/eurocrypt/ChabaudV94}}
Let $F \colon \ftwo^n \rightarrow \ftwo^n$. Then the following two statements are equivalent:
\begin{enumerate}
    \item $F$ is APN.
    \item $\sum_{a,b \in \ftwo^n, b \neq 0} \widehat{F}_b^4(a) = 2^{4n+1}-2^{3n+1}$.
\end{enumerate}
\end{lemma}

\begin{theorem}
Let $n \in \mathbb{N},n\geq 3$ be odd and let $F \colon \ftwo^{n+1} \rightarrow \ftwo^{n+1}$ be a quadratic APN function with linearity $2^n$. Then, $F$ consists of $2^{n-1}$ components whose Walsh transform only take values in $\{ 0, \pm 2^{\frac{n+3}{2}}\}$ (i.e., semi-bent components), $2^{n+1}-2-2^{n-1}$ components whose Walsh transform only take values in $\{\pm 2^{\frac{n+1}{2}} \}$ (i.e., bent components), and a single component with linearity $2^n$. 
\end{theorem}
\begin{proof}
Let $T \colon \ftwo^n \times \ftwo \rightarrow \ftwo^n \times \ftwo$ be a function EA-equivalent to $F$ and in the form as given in Theorem~\ref{thm:canonical_form}. For $a,b \in \ftwo^n$ and $\bar{a},\bar{b} \in \ftwo$, it is straightforward to deduce that
\begin{align*}
    \widehat{T}_{(b,\bar{b})}(a,\bar{a}) = \begin{cases} \widehat{G}_b(a) + (-1)^{\bar{a}}\widehat{G}_b(a+b) & \text{ if } \bar{b}=0\\
    \widehat{G}_b(a) + (-1)^{\bar{a}}\widehat{G}_b(a+b+\gamma) & \text{ if } \bar{b}=1\;.
    \end{cases}
\end{align*}

If $b = 0$, we only have one component $(b,\bar{b}) = (0,1)$. In this case, we obtain $\widehat{T}_{(0,\bar{1})}(a,\bar{a}) = \widehat{G}_0(a) + (-1)^{\bar{a}}\widehat{G}_0(a+\gamma)$, which evaluates to 0 if and only if $(a,\bar{a}) \notin \{(0,0),(0,1),(\gamma,0),(\gamma,1) \}$. Otherwise, $\widehat{T}_{(0,\bar{1})}(a,\bar{a})$ evaluates to $\pm 2^{n}$.

Let us now consider the case $b \neq 0$. Since $G$ is a quadratic APN function in odd dimension $n$, it must be almost bent. In other words, its component functions are all semi-bent, i.e., for all $a \in \ftwo^n$, we have $\widehat{G}_b(a) \in \{0,\pm 2^{\frac{n+1}{2}}\}$. Therefore, each component function $(b,\bar{b})$ of $T$ with $b \neq 0$ is either bent (i.e., $\forall (a,\bar{a}) \in \ftwo^n \times \ftwo \colon \widehat{T}_{(b,\bar{b})}(a,\bar{a}) \in \{\pm 2^{\frac{n+1}{2}}\}$) or semi-bent (i.e., $\forall (a,\bar{a}) \in \ftwo^n \times \ftwo \colon \widehat{T}_{(b,\bar{b})}(a,\bar{a}) \in \{0,\pm 2^{\frac{n+3}{2}}\}$).

Let $k$ denote the number of semi-bent components of $F$ and let the sets $A,B,C$ be defined as follows:
\begin{align*}
    A &\coloneqq \{(a,b) \in \ftwo^{n+1} \times \ftwo^{n+1} \setminus \{0\} \mid \lvert\widehat{F}_{b}(a)\rvert = 2^n\} \\
    B &\coloneqq \{(a,b) \in \ftwo^{n+1} \times \ftwo^{n+1} \setminus \{0\} \mid \lvert\widehat{F}_{b}(a)\rvert = 2^{\frac{n+1}{2}}\} \\
    C &\coloneqq \{(a,b) \in \ftwo^{n+1} \times \ftwo^{n+1} \setminus \{0\} \mid \lvert\widehat{F}_{b}(a)\rvert = 2^{\frac{n+3}{2}}\}\;. \\
\end{align*}
From the previous observations, the number of bent components of $F$ is equal to $2^{n+1}-2-k$ and we have $\lvert A\rvert = 4$, $\lvert B\rvert = (2^{n+1}-2-k)2^{n+1}$, and $\lvert C\rvert = k2^{n-1}$. The cardinalities of the sets $\lvert B\rvert$ and $\lvert C\rvert$ stated as above follow from Parseval's relation for the Walsh transform, i.e., for all $b \neq 0$, we have $\sum_{a \in \ftwo^{n+1}}\widehat{F}_b^2(a) = 2^{2n+2}$ (see~\cite[p.\@ 61]{carlet_2021}).
We thus have
\begin{align*}
    \sum_{a,b \in \ftwo^{n+1}, b \neq 0} \widehat{F}_b^4(a) &= \lvert A\rvert 2^{4n} + \lvert B\rvert 2^{2n+2} + \lvert C\rvert 2^{2n+6} \\
    &= k(2^{3n+5}-2^{3n+3}) + 2^{4n+2} + 2^{4n+4} - 2^{3n+4}\;,
\end{align*}
which must be equal to $2^{4n+5}-2^{3n+4}$ according to Lemma~\ref{lem:4th_moments}. It follows that $k = 2^{n-1}$.
\end{proof}

\paragraph{On the Ortho-derivative of APN Functions with Maximum Linearity}
If $T \colon \ftwo^n \times \ftwo \rightarrow \ftwo^n \times \ftwo$ is a quadratic APN function with linearity $2^n$, then there exists an $(n-1)$-dimensional linear space $V$ such that the ortho-derivative $\pi_T$ of $T$ is constant on $V \setminus \{0\}$. Indeed, for any fixed non-zero $\gamma \in \ftwo^n$, the APN function $T$ is EA-equivalent to a function $T'$ of the form as in Theorem~\ref{thm:canonical_form}, so for all $(x,y) \in \ftwo^n \times \ftwo$, we have that $\pi_{T'}(\alpha,\beta) \neq 0$ is orthogonal to
\begin{equation*}
    \left(\begin{array}{c}
     G(x) + G(x+\alpha) + G(\alpha) + G(0)  \\
     0 
\end{array}\right) + \beta \left(\begin{array}{c}
     x  \\
     \langle \gamma, x \rangle 
\end{array}\right) + y \left(\begin{array}{c}
     \alpha  \\
     \langle \gamma,\alpha \rangle 
\end{array}\right)
\end{equation*}
as long as $(\alpha,\beta) \in (\ftwo^n \times \ftwo) \setminus \{0\}$. Setting $V' \coloneqq \{(\alpha,\beta) \in \ftwo^n \times \ftwo \mid \beta = 0 \text{ and } \langle \gamma, \alpha \rangle = 0 \}$, we can observe that $\pi_{T'}(\alpha,\beta) = ((0,0,\dots,0),1)$ for all $(\alpha,\beta) \in V' \setminus \{0\}$. We recall that the ortho-derivatives of EA-equivalent functions are linear-equivalent~\cite[Prop.\@ 39]{ortho_paper}.

\subsection{The Case of Gold Functions}
\label{sec:gold}
It is worth highlighting that an EA-equivalent representation of the 6-bit quadratic APN function with linearity $2^5$, i.e., Function no.\@ 2.6 in~\cite[Table 5]{DBLP:journals/amco/EdelP09}, can be given as
\begin{align*}
    T_6 \colon \mathbb{F}_{2^5} \times \ftwo &\rightarrow \mathbb{F}_{2^5} \times \ftwo\\
    \left(\begin{array}{c}
     x  \\
     y 
\end{array}\right) &\mapsto 
\left(\begin{array}{c}
     x^3  \\
     0 
\end{array}\right)+
\left(\begin{array}{c}
     x^{16}+x  \\
     \tr{x} 
\end{array}\right)\cdot y\;,
\end{align*}
i.e., it can be obtained as a 0-extension of the cube function over $\mathbb{F}_{2^5}$. Note that it was already observed in~\cite[Sec.\@ 3]{DBLP:journals/iacr/KalginI20a} that both of the two quadratic EA-equivalence classes of 5-bit APN functions yield the EA-equivalence class of the 6-bit quadratic APN function with maximum linearity as a 0-extension. The following proposition gives us a necessary and sufficient condition on when a Gold APN function in odd dimension can be extended to a quadratic APN function with maximum linearity.
\begin{proposition}
Let $n \in \mathbb{N}$ be an odd integer and let $i \in \mathbb{N}$ be an integer with $\gcd(i,n)=1$. The APN function $G \colon \mathbb{F}_{2^n} \rightarrow \mathbb{F}_{2^n}, x \mapsto x^{2^i+1}$ is 0-extendable if and only if there exists a linear function (bijection) $L \colon \mathbb{F}_{2^n} \rightarrow \mathbb{F}_{2^n}$ such that, for all $x \in \mathbb{F}_{2^n}\setminus\{0\}$ with $\tr{ x}=0$, we have $\tr{x^{-(2^i+1)}L(x)} = 1$.
\end{proposition}
\begin{proof}
In the finite field $\mathbb{F}_{2^n}$, we use $\langle \alpha,x \rangle_{\mathbb{F}_{2^n}} = \tr{\alpha x}$. Then, for the ortho-derivative of $G$, we have \begin{align*}\pi_G \colon \mathbb{F}_{2^n} \rightarrow \mathbb{F}_{2^n}, x \mapsto \begin{cases}x^{-(2^i+1)} &\text{ if } x \neq 0  \\ 0 &\text{ if } x = 0 \end{cases}\;.\end{align*}
Indeed, for all $\alpha,x \in \mathbb{F}_{2^n}$, we have $B_{\alpha}(x) \coloneqq G(x) + G(x+\alpha) + G(\alpha) + G(0) = \alpha x^{2^i} + \alpha^{2^i} x$ and thus, $\tr{x^{-(2^i+1)}B_{\alpha}(x)} = \tr{\alpha x^{-1}} + \tr{\alpha^{2^i}x^{-2^i}} = 0$.

Therefore, $G$ is 0-extendable if and only if there exists a linear function $L' \colon \mathbb{F}_{2^n} \rightarrow \mathbb{F}_{2^n}$ and a non-zero $\beta \in \mathbb{F}_{2^n}$ such that, for all $x \in \mathbb{F}_{2^n}\setminus\{0\}$ with $\tr{\beta x}=0$, we have $\tr{x^{-(2^i+1)}L'(x)} = 1$. If we substitute $x$ by $\beta^{-1} x$, we obtain that, for all $x \in \mathbb{F}_{2^n}\setminus\{0\}$ with $\tr{x}=0$, we have $\tr{x^{-(2^i+1)}\beta^{2^i+1}L'(\beta^{-1}x)} = 1$. The result follows by defining $L \colon \mathbb{F}_{2^n} \rightarrow \mathbb{F}_{2^n}, x \mapsto \beta^{2^i+1}L'(\beta^{-1}x)$.

Without loss of generality, we can assume $L$ being a bijection because of Proposition~\ref{prop:many_L}.
\end{proof}

We did not find any example of a 0-extendable Gold function in odd dimension $n$ with $7 \leq n \leq 15$.

\section{Open Problems}
Our work leaves several open problems, which we list in the following. We expect that a solution to any of those problems will provide further interesting insights within the theory of APN functions.

\begin{openproblem}
Study how restrictions of a function $F \colon \ftwo^n \rightarrow \ftwo^n$ are related to restrictions of a function $F'$ CCZ-equivalent to $F$. Moreover, if we have two quadratic APN functions $F$ and $G$ such that $G \prec F$, determine whether it is possible to construct an APN function $G'$ which is not CCZ-equivalent to a quadratic function such that $G' \prec F'$, where $F'$ is CCZ-equivalent to $F$.
\end{openproblem}

\begin{openproblem}
Find a recursive APN function in every dimension $n \geq 2$ (i.e., prove Conjecture~\ref{con:recursive}).
\end{openproblem}

\begin{openproblem}
Find an algorithmic way for finding valid pairs $(L,\ell)$ when building $r$-extendable APN functions that is more efficient than the method described in Section~\ref{sec:r-extendable}.
\end{openproblem}

\begin{openproblem}
Construct an infinite family of quadratic APN functions with maximum linearity (or equivalently an infinite family of 0-extendable APN functions) or prove that such a family cannot exist.
\end{openproblem}
  
\begin{openproblem}
Determine whether there exists a 0-extendable APN function in odd dimension $n>5$ that comes from one of the known infinite families of APN functions.
\end{openproblem}

\subsubsection*{Acknowledgments}
We thank the anonymous reviewers for their detailed and helpful comments.

This  work  was  funded  by Deutsche  Forschungsgemeinschaft  (DFG);  project  number 411879806 and by DFG under Germany's Excellence Strategy - EXC 2092 CASA - 390781972.

\appendix
\section{An Example of a Recursive APN Function in Dimension Eight}
\label{app:recursive}

Below, we provide the look-up table of a quadratic recursive APN function $R \colon \ftwo^8 \rightarrow \ftwo^8$. The entries of the look-up table are denoted as hexadecimal values, omitting the \texttt{0x} prefix for presentation purposes.
\begin{align*}
    R = [
    &\texttt{00},\texttt{79},\texttt{b2},\texttt{e1},\texttt{39},\texttt{c7},\texttt{70},\texttt{a4},\texttt{36},\texttt{c0},\texttt{22},\texttt{fe},\texttt{5e},\texttt{2f},\texttt{b1},\texttt{ea},\\
    &\texttt{b9},\texttt{f8},\texttt{1d},\texttt{76},\texttt{28},\texttt{ee},\texttt{77},\texttt{9b},\texttt{0a},\texttt{c4},\texttt{08},\texttt{ec},\texttt{ca},\texttt{83},\texttt{33},\texttt{50},\\
    &\texttt{1e},\texttt{1d},\texttt{70},\texttt{59},\texttt{b5},\texttt{31},\texttt{20},\texttt{8e},\texttt{58},\texttt{d4},\texttt{90},\texttt{36},\texttt{a2},\texttt{a9},\texttt{91},\texttt{b0},\\
    &\texttt{8d},\texttt{b6},\texttt{f5},\texttt{e4},\texttt{8e},\texttt{32},\texttt{0d},\texttt{9b},\texttt{4e},\texttt{fa},\texttt{90},\texttt{0e},\texttt{1c},\texttt{2f},\texttt{39},\texttt{20},\\
    &\texttt{8e},\texttt{26},\texttt{1f},\texttt{9d},\texttt{ba},\texttt{95},\texttt{d0},\texttt{d5},\texttt{a6},\texttt{81},\texttt{91},\texttt{9c},\texttt{c3},\texttt{63},\texttt{0f},\texttt{85},\\
    &\texttt{fc},\texttt{6c},\texttt{7b},\texttt{c1},\texttt{60},\texttt{77},\texttt{1c},\texttt{21},\texttt{51},\texttt{4e},\texttt{70},\texttt{45},\texttt{9c},\texttt{04},\texttt{46},\texttt{f4},\\
    &\texttt{2f},\texttt{fd},\texttt{62},\texttt{9a},\texttt{89},\texttt{dc},\texttt{3f},\texttt{40},\texttt{77},\texttt{2a},\texttt{9c},\texttt{eb},\texttt{80},\texttt{5a},\texttt{90},\texttt{60},\\
    &\texttt{77},\texttt{9d},\texttt{2c},\texttt{ec},\texttt{79},\texttt{14},\texttt{d9},\texttt{9e},\texttt{aa},\texttt{cf},\texttt{57},\texttt{18},\texttt{f5},\texttt{17},\texttt{f3},\texttt{3b},\\
    &\texttt{46},\texttt{bf},\texttt{d8},\texttt{0b},\texttt{0b},\texttt{75},\texttt{6e},\texttt{3a},\texttt{9c},\texttt{ea},\texttt{a4},\texttt{f8},\texttt{80},\texttt{71},\texttt{43},\texttt{98},\\
    &\texttt{eb},\texttt{2a},\texttt{63},\texttt{88},\texttt{0e},\texttt{48},\texttt{7d},\texttt{11},\texttt{b4},\texttt{fa},\texttt{9a},\texttt{fe},\texttt{00},\texttt{c9},\texttt{d5},\texttt{36},\\
    &\texttt{ef},\texttt{6c},\texttt{ad},\texttt{04},\texttt{30},\texttt{34},\texttt{89},\texttt{a7},\texttt{45},\texttt{49},\texttt{a1},\texttt{87},\texttt{cb},\texttt{40},\texttt{d4},\texttt{75},\\
    &\texttt{68},\texttt{d3},\texttt{3c},\texttt{ad},\texttt{1f},\texttt{23},\texttt{b0},\texttt{a6},\texttt{47},\texttt{73},\texttt{b5},\texttt{ab},\texttt{61},\texttt{d2},\texttt{68},\texttt{f1},\\
    &\texttt{d1},\texttt{f9},\texttt{6c},\texttt{6e},\texttt{91},\texttt{3e},\texttt{d7},\texttt{52},\texttt{15},\texttt{b2},\texttt{0e},\texttt{83},\texttt{04},\texttt{24},\texttt{e4},\texttt{ee},\\
    &\texttt{b7},\texttt{a7},\texttt{1c},\texttt{26},\texttt{5f},\texttt{c8},\texttt{0f},\texttt{b2},\texttt{f6},\texttt{69},\texttt{fb},\texttt{4e},\texttt{4f},\texttt{57},\texttt{b9},\texttt{8b},\\
    &\texttt{c7},\texttt{95},\texttt{a6},\texttt{de},\texttt{15},\texttt{c0},\texttt{8f},\texttt{70},\texttt{73},\texttt{ae},\texttt{b4},\texttt{43},\texttt{f0},\texttt{aa},\texttt{cc},\texttt{bc},\\
    &\texttt{8b},\texttt{e1},\texttt{fc},\texttt{bc},\texttt{f1},\texttt{1c},\texttt{7d},\texttt{ba},\texttt{ba},\texttt{5f},\texttt{6b},\texttt{a4},\texttt{91},\texttt{f3},\texttt{bb},\texttt{f3}
    ]
\end{align*}

Let $g \in \F_{2^8}$ be an element with minimal polynomial $X^8 + X^4 + X^3 + X^2 + 1$. Then, $R$ can also be given as a function over $\F_{2^8}$ by the univariate representation
\begin{equation*}
    \begin{split}
    R(x) &= 
    g^{157} x + g^{237} x^{2} + g^{169} x^{3} + g^{56} x^{4} + g^{43} x^{5} + g^{11} x^{6} + g^{154} x^{8} + g^{89} x^{9} \\ &+ g^{155} x^{10} + g^{221} x^{12} + g^{157} x^{16} + g^{5} x^{17} + g^{245} x^{18} + g^{32} x^{20} + g^{127} x^{24} \\ &+ g^{49} x^{32} + g^{81} x^{33} + g^{4} x^{34} + g^{146} x^{36} + g^{223} x^{40} + g^{44} x^{48} + g^{70} x^{64} \\ &+ g^{127} x^{65} + g^{113} x^{66} + g^{52} x^{68} + g^{253} x^{72} + g^{209} x^{80} + g^{239} x^{96} + g^{43} x^{128} \\ &+ g^{4} x^{129} + g^{99} x^{130} + g^{89} x^{132} + g^{26} x^{136} + g^{47} x^{144} + g^{220} x^{160} + g^{253} x^{192}.
\end{split}
\end{equation*}

The following sage code returns $R$ and APN restrictions contained in $R$ for $n=7,\dots,2$. 
\begin{lstlisting}[frame=single, basicstyle=\small, tabsize=2, language=python]
n = 8
while (n>=2):
    mask = sum(int(1 << i) for i in xrange(0,n))
    R = [R[x] & mask for x in xrange(0, 2**n)] 
    print(R)
    n = n-1
\end{lstlisting}

\section{Sage Implementation of Algorithm~\ref{alg:classification}}
\label{app:sage}
\begin{lstlisting}[frame=single, basicstyle=\small, tabsize=2, language=python]
from sboxU import *
load("sboxU/known_functions/sevenBitAPN.py")
n = 7

AllQuadG = all_quadratics()

def bit(a,i):
	return ((a>>i)%2)

def inner_product(l,x):
	return Integer(l&x).bits().count(1)%2

# converts the matrix A to a look-up table
def matrix_to_sbox(A):
	T = []
	V = VectorSpace(GF(2),A.nrows())
	for v in V:
		T.append(ZZ(list(A*v),base=2))
	return T

def gen_system(G,ortho,l):
	M = []
	for a in range(2**n)[1::]:
		if (inner_product(l,a)==0):
			row = []
			for bit_pi in range(n):
				for bit_a in range(n):
					row.append(bit(a,bit_a)*bit(ortho[a],bit_pi))
			M.append(row)
	return (matrix(GF(2),M))

def is_extendable(G,ortho,l):
	M = gen_system(G,ortho,l)
	v = vector(GF(2),[1]*(M.nrows()))
	try:
		mat = M.solve_right(v)
	except (ValueError):
		return []
	if (M.right_kernel().dimension()>(2*n)):
		print("There might be more solutions!\n")
	#otherwise there is only on L up to Gamma-equivalence
	L = matrix_to_sbox(matrix(GF(2),n,n,list(mat)))		
	S = []
	for i in range(2**n):
		S.append(int(G[i]<<1))
	for i in range(2**n):
		S.append(int(((G[i]^^L[i])<<1)^^inner_product(l,i)))
	return S

# returns a list of APN extensions of G
def extensions(G):
	sol = []
	pi = ortho_derivative(G)
	for l in range(2**n):
		t = is_extendable(G,pi,l)
		if (not(t==[])):
			sol.append(t)
	return sol
	
sol = []
for i in range(len(AllQuadG)):
	print(i)
	t = extensions(AllQuadG[i])
	if (not(t==[])):
		sol.append(t)
		
print(sol)
\end{lstlisting}

\end{document}